\documentclass[12pt]{article}

\usepackage[margin=1.25in]{geometry}

\usepackage{setspace}
\usepackage{amsmath,amsthm,amssymb,mathtools}

\usepackage{graphicx}
\usepackage{dsfont}
\usepackage{tikz}
\usetikzlibrary{calc}
\usetikzlibrary{shapes.geometric}
\usepackage{paralist}

\usepackage{wrapfig}

\usepackage{float}
\usepackage{comment}
\usepackage{algorithm}
\usepackage[noend]{algpseudocode}

\usepackage{thmtools}
\usepackage{thm-restate}

\usepackage{caption}
\usepackage[labelformat=simple]{subcaption}

\usepackage{aliascnt}

\usepackage{xcolor}
\usepackage{colortbl}

\usepackage[bottom]{footmisc}

\PassOptionsToPackage{hyphens}{url}
\usepackage{hyperref}
\hypersetup{colorlinks=true,citecolor=blue,linkcolor=blue}
\hypersetup{colorlinks=true}
\hypersetup{linkcolor=[rgb]{.7,0,0}}
\hypersetup{citecolor=[rgb]{0,.7,0}}
\hypersetup{urlcolor=[rgb]{.7,0,.7}}

\declaretheorem[numberwithin=section]{theorem}
\declaretheorem[sibling=theorem,style=definition]{definition}
\declaretheorem[sibling=theorem]{lemma}

\declaretheorem[sibling=theorem,style=definition]{remark}
\declaretheorem[sibling=theorem]{proposition}
\declaretheorem[sibling=theorem,style=definition]{example}

\usepackage{rotating}
\usepackage{makecell}
\usepackage[round]{natbib}
\usepackage{booktabs}
\usepackage{multirow}

\usepackage{float}
\usepackage{comment}
\usepackage{soul}
\usepackage{wrapfig}

\usepackage{algorithmicx}
\usepackage[noend]{algpseudocode}
\usepackage{algpseudocodex} 
\usepackage{extdash}

\hypersetup{pdfauthor={Yannai A. Gonczarowski, Ori Heffetz, and Clayton Thomas}}
\hypersetup{pdftitle={Strategyproofness-Exposing Descriptions of Matching Mechanisms}}

\makeatletter
\renewcommand{\ALG@name}{Description}
\makeatother

\algdef{SE}[DOWHILE]{Do}{doWhile}{\algorithmicdo}[1]{\algorithmicwhile\ #1}%

\newcommand{\BigO}{\mathcal{O}}

\newcommand{\E}[1]{\mathds{E}\left[{#1}\right]}

\makeatletter
\newcommand\Ps@textstyle[2]{\mathbb{P}_{#1}\left[{#2}\right]}
\newcommand\Es@textstyle[2]{\mathbb{E}_{#1}\left[{#2}\right]}
\newcommand\Ps[2]{%
  \mathchoice
  {\underset{{#1}}{\mathbb{P}}\left[{#2}\right]}
  {\Ps@textstyle{#1}{#2}}
  {\Ps@textstyle{#1}{#2}}
  {\Ps@textstyle{#1}{#2}}
}
\newcommand\Es[2]{%
  \mathchoice
  {\underset{{#1}}{\mathbb{E}}\left[{#2}\right]}
  {\Es@textstyle{#1}{#2}}{\Es@textstyle{#1}{#2}}{\Es@textstyle{#1}{#2}}
}
\makeatother

\usepackage{comment}
\usepackage{adjustbox}

\usepackage{etoolbox}
\makeatletter
\patchcmd{\hyper@makecurrent}{%
    \ifx\Hy@param\Hy@chapterstring
        \let\Hy@param\Hy@chapapp
    \fi
}{%
    \iftoggle{inappendix}{%
        \@checkappendixparam{chapter}%
        \@checkappendixparam{section}%
        \@checkappendixparam{subsection}%
        \@checkappendixparam{subsubsection}%
        \@checkappendixparam{paragraph}%
        \@checkappendixparam{subparagraph}%
    }{}%
}{}{\errmessage{failed to patch}}

\newcommand*{\@checkappendixparam}[1]{%
    \def\@checkappendixparamtmp{#1}%
    \ifx\Hy@param\@checkappendixparamtmp
        \let\Hy@param\Hy@appendixstring
    \fi
}
\makeatletter

\newtoggle{inappendix}
\togglefalse{inappendix}

\apptocmd{\appendix}{\toggletrue{inappendix}}{}{\errmessage{failed to patch}}

\usepackage{sidecap}

\usepackage{csquotes}

\usepackage{tikz}
\usetikzlibrary{matrix}
\usetikzlibrary{positioning}
\usetikzlibrary{arrows, decorations.pathmorphing}
\usetikzlibrary{decorations.pathreplacing,calligraphy}

\newcommand{\Menu}{\mathcal{M}}

\newcommand{\D}{\mathcal{D}}
\renewcommand{\H}{\mathcal{H}}
\newcommand{\T}{\mathcal{T}}

\begin{document}
\title{
\vspace*{-3em}
Strategyproofness-Exposing Descriptions of Matching Mechanisms%
\thanks{%
An earlier version of this paper circulated under the name ``Strategyproofness-Exposing Mechanism Descriptions,'' a one-page abstract of which appeared in EC 2023. 
For helpful comments and discussions, the authors thank Itai Ashlagi, Tilman B{\"o}rgers, Eric Budish, Robert Day,
Benjamin Enke, Xavier Gabaix, Nicole Immorlica, Guy Ishai, David Laibson, Jacob Leshno, Irene Lo, Kevin Leyton-Brown, Shengwu Li, Paul Milgrom, Michael Ostrovsky, Assaf Romm, Al Roth, Fedor Sandomirskiy, Ilya Segal,
Ran Shorrer, Pete Troyan, Matt Weinberg, Leeat Yariv, Huacheng Yu, 
and seminar participants at 
HUJI, Cornell, Harvard, Stanford GSB, WALE 2022, Noam Nisan's 60th Birthday Conference, the INFORMS Workshop at EC 2022, MATCH-UP 2022, University of Michigan, INFORMS 2022, the 2022 NBER Market Design Meeting, University of Tokyo, UCI, UCSD, the Bellairs Workshop 2023, EC 2023, the 2023 WZB Matching Market Design conference, Microsoft Research, University of Kentucky, Northwestern, and RPI.
Zo{\"e} Hitzig provided valuable research assistance.
Gonczarowski gratefully acknowledges research support by the National Science Foundation (NSF-BSF grant No.\ 2343922), Harvard FAS Inequality in America Initiative, and Harvard FAS Dean’s Competitive Fund
for Promising Scholarship; parts of his work were carried out while at Microsoft Research and while vising MIT CSAIL and MIT Sloan. 
Heffetz gratefully acknowledges research support by the Israel Science Foundation (grant No.\ 2968/21), US-Israel Binational Science Foundation (NSF-BSF grant No.\ 2023676), and Cornell's Johnson School; part of his work was carried out while visiting Princeton's School of Public and International Affairs. 
Thomas gratefully acknowledges the support of NSF CCF-1955205, a Wallace Memorial Fellowship in Engineering, and a Siebel Scholar award; parts of his work were carried out while at Princeton, Microsoft Research, and Yale.
} 
}
\author{
\vspace{-0.2in}
  Yannai A. Gonczarowski\thanks{Department of Economics and Department of Computer Science, Harvard University | \emph{E-mail}: \href{mailto:yannai@gonch.name}{yannai@gonch.name}.}
  \and
  Ori Heffetz\thanks{Johnson Graduate School of Management, Cornell University, Bogen Department of Economics and Federmann Center for Rationality, The Hebrew University of Jerusalem, and NBER | \emph{E-mail}: \href{mailto:oh33@cornell.edu}{oh33@cornell.edu}.}
  \and
  Clayton Thomas\thanks{Department of Computer Science, Rensselaer Polytechnic Institute | \emph{E-mail}: \href{mailto:thomas.clay95@gmail.com}{thomas.clay95@gmail.com}.}
}

\date{September 9, 2026}

\maketitle

\vspace{-2.5em}

\begin{abstract}

A \emph{menu description} exposes strategyproofness by presenting a mechanism to player $i$ in two steps.
Step~(1) uses others' reports to describe $i$'s \emph{menu} of potential outcomes.
Step~(2) uses $i$'s report to select $i$'s favorite outcome from her menu.
We provide novel menu descriptions of the Deferred Acceptance (DA) and Top Trading Cycles (TTC) matching mechanisms.
For TTC, our description additionally yields a proof of the strategyproofness of TTC's traditional description, in a way that we prove is impossible for DA.

\end{abstract}

\thispagestyle{empty}

\clearpage

\setcounter{page}{1}

\section{Introduction}
\label{sec:introduction}

\vspace{-0.07in}

Strategyproof mechanisms are often considered desirable.
Under standard economic assumptions, these mechanisms eliminate the need for players to strategize, since straightforward play is a dominant strategy.\footnote{
  We use the term ``straightforward'' to describe the strategy an agent would play under classic economic assumptions. While often referred to as the ``truthtelling'' strategy, we avoid this morally laden term, since deviations from this strategy should not be thought of as dishonesty.
} 
In practice, however, real participants in strategyproof mechanisms often play theoretically dominated strategies, raising the possibility that they do not perceive the mechanisms as strategyproof (see, e.g., \citealt{HakimovK19, Rees-JonesS23}).

In this paper, we posit that the way mechanisms are \emph{described} can influence the extent to which participants perceive strategyproofness. 
In contrast to other recent works, which have sought to encourage straightforward play by implementing a given choice rule through different interactive mechanisms,\footnote{
  Prior works often consider interactive mechanisms designed to reduce non-straightforward play arising from behavioral factors, such as contingent-reasoning failures \citep{Li17,PyciaT21} and loss aversion \citep{DreyfussHR22,MeisnerW23}.
}
we propose changing only the (ex ante) description of a given static, direct-revelation mechanism.
We propose a general outline---called \emph{menu descriptions}---for describing mechanisms to one player at a time in a way that makes strategyproofness hold via an elementary, one-sentence argument. 
In this sense, menu descriptions {expose strategyproofness}.

Our focus is on matching, and particularly on two canonical mechanisms: Deferred Acceptance (henceforth DA) and Top Trading Cycles (henceforth TTC).
These mechanisms are widely and successfully deployed.
They are typically %
explained to participants using \emph{outcome descriptions}---i.e., detailed and explicit algorithms for calculating the outcome matching.
However, the traditional (outcome) descriptions of these mechanisms do not expose their strategyproofness,
in the sense that proving this property from these descriptions requires technical mathematical arguments.

We present three results.
The first is a new menu description of DA.
The second is a new menu description of TTC, which furthermore yields a new proof of the strategyproofness of the traditional description of TTC.
The third is an impossibility result showing that such a proof via a menu description cannot work for the traditional description of DA.

As an initial illustration, consider the canonical Serial Dictatorship (henceforth SD) mechanism.
When matching applicants to institutions,\footnote{
  Throughout this paper, the only strategic players are the applicants.
  Institutions are non-strategic, and their preferences over the applicants are by convention called \emph{priorities}.
  \vspace{-0.1in}
}
the traditional description of SD is as follows: In some priority order, say $i=1,\ldots,n$, applicant $i$ is matched to her highest-ranked not-yet-matched institution.
Strategyproofness is exposed by this description:
Applicant $i$ cannot influence the set of not-yet-matched institutions, and straightforward reporting guarantees $i$ her favorite not-yet-matched institution.
Our paper presents new descriptions of DA and TTC that make strategyproofness as evident as in the traditional description of SD.

In \autoref{sec:prelims}, we provide preliminaries.
We study descriptions in terms of the classic notion of a \emph{menu} \citep{Hammond79}---the set of all institutions an applicant might match to, given others' reports.
In particular, a \emph{menu description} for applicant~$i$ has the following two-step outline:
\begin{itemize}[]
  \label{itemize:menu-description}
\vspace{-0.5em}
  \setlength{\itemsep}{0em}
  \item \textbf{Step (1)} uses only the reports of other applicants to describe $i$'s menu. 
  \item \textbf{Step (2)} says that $i$'s match is her highest-ranked institution from her menu.
\end{itemize}
\vspace{-0.5em}
For instance, the traditional description of SD is a menu description.
In contrast with some other mechanisms' traditional descriptions, strategyproofness is exposed by any menu description in the same way as in SD:
In Step~(1), applicant $i$ cannot influence her menu (in SD, the set of not-yet-matched institutions), and in Step~(2), straightforward reporting guarantees $i$ her favorite %
institution from her menu.

\begin{table}[b]
    \caption[Menu DA]{
      Two %
      descriptions of DA (the applicant-optimal stable match)
      }
    \label{fig:example-descriptions-DA}
    {\centering
    \begin{tabular}{|l|l|}
    \hline
        \begin{minipage}{0.25\textwidth}
        \raggedright
        \vspace{0.05in}
        \underline{Traditional Descr.\vphantom{p}:}

        The applicants and institutions will be matched using the \emph{applicant-proposing Deferred~Acceptance} algorithm. 
        \vspace{0.05in}
        \end{minipage}
         &
        \begin{minipage}{0.7\textwidth}
        \vspace{0.05in}
        \underline{Menu Description:}

        We will run \emph{institution-proposing Deferred Acceptance} with all applicants \emph{except you}, to obtain a hypothetical matching.
        Your menu consists of every institution that ranks you higher than its hypothetically matched applicant. 

        You will be matched to the institution that you \emph{ranked highest} out of your menu. 
        \vspace{0.05in}
        \end{minipage}
        \\    \hline
    \end{tabular}
    }
    
     {\footnotesize \textbf{Note:}
     In the menu description, others' hypothetical matches need not be their matches in DA. \par
     \par
     }
\end{table}

In \autoref{sec:parallen-hpda}, we present our first result: A novel menu description of DA.
Our description---which is summarized in \autoref{fig:example-descriptions-DA}---describes applicant $i$'s menu as all institutions that prefer~$i$ to their outcome in ``flipped-side-proposing'' Deferred Acceptance without $i$.
This directly conveys $i$'s match in DA, while exposing strategyproofness for $i$. 
Prior to our work, it was not clear how to construct DA's menu, except via a trivial brute-force solution (applicable for any strategyproof mechanism) involving running the traditional description many times to separately check whether or not each institution is on $i$'s menu.

In \autoref{sec:ttc-individ-dict}, we present our second result: A novel menu description of TTC, which furthermore yields a novel proof of the strategyproofness of TTC's traditional description.
Crucial to this proof is that our menu description is naturally contained in an outcome description that is a slight tweak of TTC's traditional description. 
In particular, this outcome description---which is summarized in \autoref{fig:example-descriptions-TTC}---has the following three-step outline (whose first two steps are our menu description):
\begin{enumerate}[]
\vspace{-0.5em}
  \itemsep0em
  \item \textbf{Step (1)} uses only the reports of other applicants to describe $i$'s menu. 
  \item \textbf{Step (2)} says that $i$'s match is her highest-ranked institution from her menu.
  \item \textbf{Step (3)} describes the rest of the matching (for all other applicants).%
  \vspace{-0.25em}
\end{enumerate}

TTC's traditional description works in terms of ``eliminating trading cycles,'' and it is well known that these cycles can be eliminated in any order. 
Our three-step description of TTC is a slight tweak of the traditional one: 
It differs only by changing the order in which cycles are eliminated (by eliminating the cycle involving applicant $i$ as late as possible).
Thus, for any applicant $i$, the match of $i$ in the traditional description equals her match in our menu description (the first two steps of our three-step description)---which (like all menu descriptions) is strategyproof.
This proves that TTC's traditional description is strategyproof.\footnote{
    \label{footnote:morrilr}
    Following the first appearance of our paper, the survey article of \cite{MorrillR24} adopted our proof of TTC's (traditional description's) strategyproofness.
    Regarding the potential real-world adoption of TTC for public school choice, \citeauthor{MorrillR24} write: 
    \begin{quote}
    ``Our experience [\ldots]\ taught us that when we worked with school districts, we should help design not just a mechanism, but also the communication package that explained that mechanism [\ldots]. Perhaps if we had already known of the proof of [Gonczarowski, Heffetz, and Thomas] we could have explained [TTC's strategyproofness] more clearly.''
    \end{quote}
\vspace{-0.3in}
}

\begin{table}[b]
    \caption[Menu DA]{
      Two %
      descriptions of TTC
      }
    \label{fig:example-descriptions-TTC}
    {\vspace{-.4em}\centering
    \begin{tabular}{|l|l|}
    \hline
        \begin{minipage}{0.22\textwidth}
        \raggedright
        \vspace{0.05in}
        \underline{Traditional Descr.\vphantom{p}:}

        As long as \emph{trading cycles} remain: Eliminate (i.e., match according to) some trading cycle, repeat.
        \vspace{0.05in}
        \end{minipage}
         &
        \begin{minipage}{0.73\textwidth}
        \vspace{0.05in}
        \underline{Menu Description [and Outcome Description Containing It]:}

        As long as \emph{trading cycles not containing you} (and hence out of your control) remain: Eliminate such a trading cycle, repeat.

        You will be matched to the institution that you \emph{ranked highest} out of those not already matched.

        [To continue calculating the rest of the matching, as long as \emph{trading cycles} remain: Eliminate some trading cycle, repeat.]
        \vspace{0.05in}
        \end{minipage}
        \\    \hline
    \end{tabular}
    }
    
     {\footnotesize \textbf{Note:}
     The menu description with the addition of the part in square brackets is an outcome description: others' matches \emph{are} their matches in TTC. \par
     \par
     }
     \vspace{-.2em}
\end{table}

In \autoref{sec:mtr-matching}, we present our third result. We ask: Like in our result for TTC, is there a menu description of DA that yields a proof of the strategyproofness of its traditional description through a slight tweak of it? 
We present an impossibility theorem showing that the answer is \emph{no}, and demonstrating how one might reason about such questions (which might not be clear ex ante).

To establish a formal notion of a ``slight tweak,'' we consider three salient properties of DA's traditional description:
\begin{itemize}   %
  \vspace{-0.5em}
  \itemsep0em
  \item It calculates the entire outcome matching; i.e., it is an outcome description.
  \item It looks at applicants' preferences in favorite-to-least-favorite order; we say that such descriptions are \emph{applicant-proposing}.
  \item It maintains bookkeeping by only tracking and iteratively modifying a single tentative matching. Formally, its state complexity \citep{rubinstein1986finite,abreu1988structure,oprea2020whatmakes} is at most the number of possible matchings. Descriptions that have this order of magnitude of state complexity are called \emph{linear-memory} in computer science terms.%
  \vspace{-0.5em}
\end{itemize}
We formalize our impossibility theorem by treating the satisfaction of these three properties as a prerequisite for a description to qualify as a slight tweak of DA's traditional description. Indeed, not satisfying the first property means reaching a different final result than the traditional description, and not satisfying the second or third properties means having step-by-step calculations that are very different from those of the traditional one.\footnote{
  For instance, our three-step description of TTC (\autoref{sec:ttc-individ-dict}) is, like its traditional description, a linear-memory applicant-proposing outcome description.
  The same is not true for our menu description of DA (\autoref{sec:parallen-hpda}): It is \emph{not} an outcome description and \emph{not} applicant-proposing. And indeed, it is far from apparent (without our proof of this fact) that the output of our menu description coincides with the output of DA's traditional description.
}

Our third result is that, in contrast to (SD and) TTC, no slight tweak of DA's traditional description contains a menu description.
Concretely, we prove that no outcome description of DA  that contains a menu description (i.e., that has the same three-step outline as our TTC description) is applicant-proposing and linear-memory.\footnote{
  Maintaining bookkeeping by iteratively modifying a single tentative matching is a special case of having state complexity at most the number of possible matchings, which is in turn a special case of having linear memory. Thus, linear memory, which is the property that we use, is the most general (weakest) of these properties, and so using this property in our impossibility theorem makes this theorem stronger than had we used any of these special cases.
}
In fact, we prove a strong impossibility: Any applicant-proposing outcome description of DA that contains a menu description must use \emph{quadratic} memory---an amount of bookkeeping far greater than that of any description tracking only a single tentative matching. Such a description must have state complexity of the order of magnitude not of the number of possible matchings, but rather of the number of possible \emph{\mbox{$n$-tuples} of matchings}, where $n$ is the number of applicants in the market. This is a far greater number of states, comparable to that of a description that maintains bookkeeping by tracking and iteratively modifying $n$ matchings in parallel.
In light of our impossibility result, our approach above that exposes the strategyproofness of TTC's traditional description cannot work for DA.

\begin{figure}[b!]
    \caption{Illustration of trichotomy for traditional descriptions of SD, TTC, and DA}
    \label{fig:tri-level-distinction}
    \centering
    \begin{minipage}[t]{0.28\textwidth}
        \includegraphics[width=\textwidth]{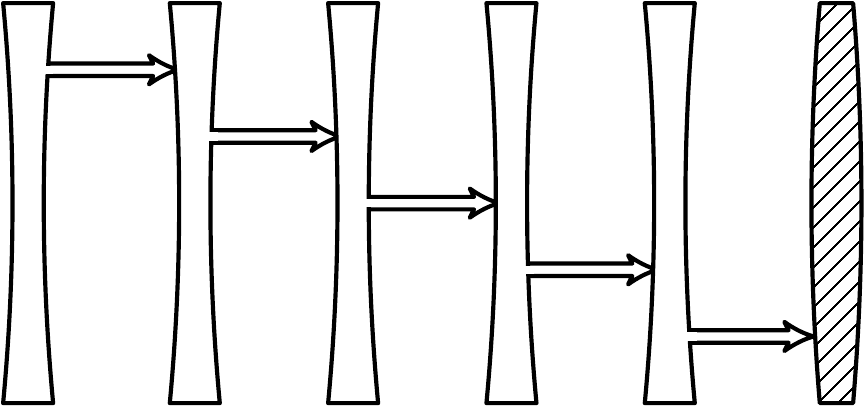}
        \subcaption{SD} %
        \label{fig:tri-SD}
    \end{minipage}
    \qquad
    \begin{minipage}[t]{0.28\textwidth}
        \includegraphics[width=\textwidth]{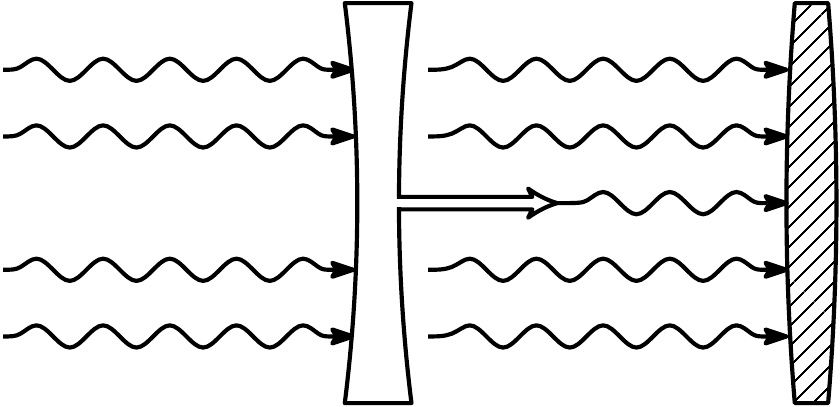}
        \subcaption{TTC} %
        \label{fig:tri-TTC}
    \end{minipage}
    \qquad
    \begin{minipage}[t]{0.28\textwidth}
        \includegraphics[width=\textwidth]{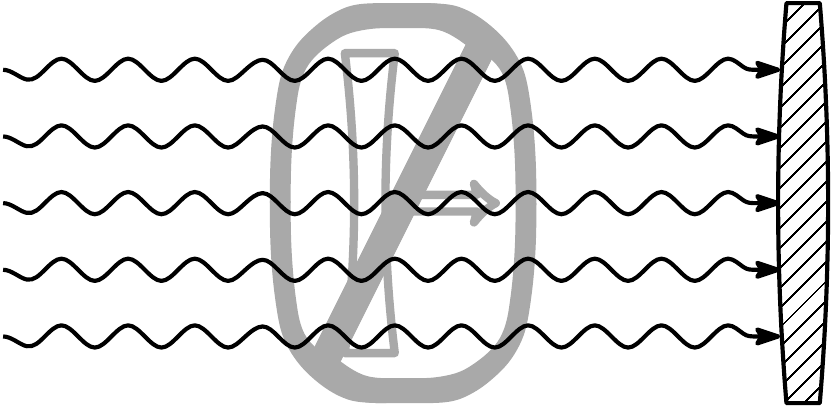}
        \subcaption{DA} %
        \label{fig:tri-DA}
    \end{minipage}
   
   \begin{minipage}[t]{\textwidth}
   { \footnotesize \textbf{Notes:}
   Each figure depicts an applicant-proposing and linear-memory outcome description, which progresses from left to right.
   The outcome matching is depicted as a shaded shape. 
   Step~(2) of a menu description (choice from a menu, which exposes strategyproofness for the chooser) is depicted as a white shape with an arrow.
   Other calculations are depicted as wavy arrows.
   The gray diagram in Panel (c) depicts the impossibility: DA cannot be described in finer detail as in Panels (a) and~(b).
   \par
   }
   \end{minipage}
\end{figure}

Our results reveal a stark trichotomy among SD, TTC, and DA.
The traditional description of SD already has each applicant choosing from her menu (see \autoref{fig:tri-SD} for an illustration), and thus exposes strategyproofness.
For the traditional description of TTC, this is not the case, but for each applicant there is a slight tweak in which she chooses from her menu (see \autoref{fig:tri-TTC}), which thus exposes strategyproofness.
For the traditional description of DA, a comparable result is impossible, in the strong sense discussed above (see \autoref{fig:tri-DA}).

\autoref{table:more-decomposed-results-table} summarizes our three results discussed above.
In \autoref{sec:related-work}, we review related work, including our empirical companion paper investigating our menu description of DA, and our work exploring menu descriptions in non-matching settings such as auctions and voting.
We conclude in \autoref{sec:conclusion}, where we also discuss potential practical concerns.

\begin{table}[h!]
  \caption{Summary of results}
  \label{table:more-decomposed-results-table}
  \vspace{-0.2in}
  \begingroup
  \begin{center}
  \center
  {\footnotesize
  \begin{tabular}{lllll}
  \toprule
  Result
  & Summary 
  & Relevant Formal Properties  
  \\  \midrule
  \makecell[l]{
    Positive result 
    \\ for \textbf{DA} 
    (Sec.\ \ref{sec:parallen-hpda}) }
  & \makecell[l]{
    \autoref{alg:Phpm} is a novel description of DA
    \\ (to one applicant at a time)
    \\ that exposes strategyproofness. }
  & \makecell[l]{
    \autoref{alg:Phpm} is a
    \\ menu description of DA.}
\\[2.5em]
  \makecell[l]{
    Positive result
    \\ for \textbf{TTC}
    (Sec.\ \ref{sec:ttc-individ-dict}) }
  & \makecell[l]{
    \autoref{alg:individ-dict-ttc} satisfies the above for TTC, \\ and additionally 
    yields a novel  
    \\ proof of the strategyproofness of
    \\ the traditional description of TTC. }
  & \makecell[l]{
    \autoref{alg:individ-dict-ttc} of TTC is a
    \\ linear-memory,
    \\ applicant-proposing 
    \\ menu-in-outcome description. } 
\\[2.5em]
  \makecell[l]{
    Negative result 
    \\ for \textbf{DA} 
    (Sec.\ \ref{sec:mtr-matching}) }
  & \makecell[l]{
    \autoref{thrm:MTR-DA-LB} shows it is impossible
    \\ to prove the strategyproofness of
    \\ the traditional description of DA
    \\ in the way we do for TTC (in Sec.\ \ref{sec:ttc-individ-dict}). }
  & \makecell[l]{
    \autoref{thrm:MTR-DA-LB} shows that for DA, 
    \\ any applicant-proposing 
    \\ menu-in-outcome description 
    \\ requires \emph{quadratic} memory.} 
    \\ \bottomrule
  \end{tabular}
  }
  \end{center}
\endgroup

 \vspace{0.4em}{\footnotesize \textbf{Note:}
     We use the term \emph{menu-in-outcome description} for an outcome description that contains a menu description (i.e., that has the same three-step outline as our TTC description). \par
     \par
     }
\end{table}

\section{Preliminaries}
\label{sec:prelims}

\subsection{Mechanisms}
\label{sec:prelims-mechanisms}

This paper studies (static, direct-revelation) matching mechanisms.
This environment consists of $n$ applicants $\{1,\ldots,n\}$ to be matched to %
institutions.
Applicant $i$ has a strict ordinal \emph{preference} $\succ_i$ over institutions, also called her \emph{type}.
Let $\T_i$ denote the set of types of applicant $i$, and let $A$ denote the set of matchings.\footnote{
  Applicants may go unmatched, and their preference lists may be partial (indicating that they prefers to going unmatched over institutions not on her preference list).
  We also let $h_1 \succ_d h_2$ indicate that applicant $d$ prefers $h_1$ to $h_2$;
  $h_1 \succeq_d h_2$ indicate $h_1\succ_d h_2$ or $h_1 = h_2$;
  $\emptyset \succ_d h$ indicate that $d$ does not rank $h$;
  $\mu(d)$ denote the match of $d$ in $\mu$;
  $\mu(d)=\emptyset$ denote $d$ going unmatched;
  $\T_{-i}$ denote the set $\T_1\times\ldots\times\T_{i-1}\times\T_{i+1}\ldots\T_n$, 
  and for $\succ_i\in\T_i$ and $\succ_{-i}\in\T_{-i}$, we write $(\succ_i, \succ_{-i})$ for the naturally corresponding element of $\T_1\times\ldots\times\T_n$.
  }
We focus on one-to-one matching for concreteness (though our results, particularly for DA, generalize substantially; see \autoref{sec:parallen-hpda}).

The applicants \emph{report} their types to a mechanism, which determines the outcome matching. 
Formally, a \emph{mechanism}
is any mapping $f : \T_1\times\ldots\times\T_n \to A$ from the reported types of all applicants to a matching.
We focus on \emph{strategyproof} mechanisms. 
This means that for every applicant $i$, every $\succ_i, \succ_{i}'\in \T_i$, and every $\succ_{-i}\in\T_{-i}$, we have
$f_i(\succ_i, \succ_{-i}) \succeq_i f_i(\succ_i', \succ_{-i})$,
where $f_i(\succ_1,\ldots,\succ_n)$ denotes $i$'s match in $f(\succ_1,\ldots,\succ_n)$.

We study the canonical strategyproof mechanisms SD, TTC, and DA.
These mechanisms are defined with respect to \emph{priorities} of the institutions over the applicants.
Following much of the matching literature,
we assume the institutions are non-strategic; 
hence, we treat the priorities as predetermined.
SD uses a single priority order $\succ$ over all applicants; TTC and DA use a profile of priority orders $\{ \succ_h \}_h$, one for each institution $h$.

\begin{definition}[SD]
  \label{def:SD}
  For a given priority order $\succ$, Serial Dictatorship (SD) is defined as follows.
  Given the reports, applicants are considered in order of highest-to-lowest priority,  and each applicant is permanently matched to her highest-ranked not-yet-matched institution.
\end{definition}

\begin{definition}[TTC]
  \label{def:TTC}
  For a given profile of institutions' priorities $\{\succ_h\}_h$,
  Top Trading Cycles (TTC) is defined as follows.
  Given the reports, repeat the following until everyone is matched (or has exhausted their preference lists):
  Every remaining (i.e., currently unmatched) applicant ``points'' to her favorite remaining institution, and every remaining institution points to its highest priority remaining applicant.
  There must be some cycle in this directed graph (since there is only a finite number of vertices).
  Choose any such cycle and ``eliminate'' that cycle by permanently matching each applicant in the cycle to the institution she is pointing to (and removing all matched agents from consideration for later cycles).
\end{definition}

\begin{definition}[DA]
  \label{def:DA}
  For a given profile of institutions' priorities $\{\succ_h\}_h$,
  Deferred Acceptance (DA) is defined as follows.
  Given the reports, repeat the following until every applicant is matched (or has exhausted her preference list): 
  A currently unmatched applicant is chosen to ``{propose}'' to her favorite institution that has not yet ``{rejected}'' her.
  The institution then rejects every proposal except for its highest priority applicant who has proposed to it thus far.
  Rejected applicants become (currently) unmatched, while that highest priority applicant is tentatively matched to the institution. 
  At the end, the tentative allocation becomes final.
\end{definition}

Note that DA, by convention, refers to the above \emph{applicant} proposing version of the mechanism. The sides can also be flipped, which results in the \emph{institution} proposing variant of DA. When confusion might arise, we use APDA for applicant-proposing DA and IPDA for institution-proposing DA.

TTC and DA are the two canonical matching mechanisms that are priority-based and strategyproof.
TTC is Pareto-efficient for the applicants \citep{ShapleyS74}, and DA is stable \citep{GaleS62}.
It is well known the outcome of TTC is independent of the order in which cycles are eliminated %
(see \citealp{Carroll14, MorrillR24} for direct proofs of this result, which follows from \citealp{ShapleyS74, RothP77})
and that the outcome of DA is independent of the order of proposals %
\citep{DubinsMachiavelliGaleShapley81}.

\subsection{Descriptions}
\label{sec:prelims-descriptions}

This paper studies ex ante descriptions of matching mechanisms, i.e., descriptions that are given before any concrete inputs are known.
When matching markets are described in detail to participants, this is typically done by specifying a set of explicit, precise, step-by-step instructions for calculating the result, i.e., by specifying an algorithm.\footnote{
    Of course, the way a description/algorithm is actually conveyed to participants can vary.
    One common real-world approach to relaying matching algorithms is an illustrative video using an example (see, e.g., \autoref{fig:applicant-linear-graphics:trad-DA} in \autoref{sec:simple-descriptions-matching} for such a video for DA).
    In our paper, we abstract over exactly how
    the algorithm is relayed. }
Thus, we formally define a \emph{description} to be any algorithm that uses as input the preferences of the applicants and the priorities of the institutions, and calculates some result (e.g., an outcome matching).\footnote{
    Algorithms can be defined in full mathematical detail in various ways; any such definition suffices for our purposes.
    In the Supplemental \autoref{sec:model-descriptions}, we present such a definition from first principles.
}

For any mechanism $f$, an \emph{outcome descriptions} of $f$ is an algorithms that, using input ${\succ_1,\ldots,\succ_n}$ (and the priorities of institutions), outputs the outcome matching $f(\succ_1,\ldots,\allowbreak{\succ_n)}$. 
For instance, the descriptions in Definitions~\ref{def:SD} through~\ref{def:DA} are outcome descriptions for SD, TTC, and DA, respectively.
We refer to each of these outcome descriptions as the \emph{traditional description} of the corresponding mechanism.

Beyond outcome descriptions, we study two other description outlines: menu descriptions and menu-in-outcome descriptions; these are defined in Sections~\ref{sec:def-menu} and~\ref{sec:prelims-individ-dict},  respectively.

\subsection{Menus and Menu Descriptions}
\label{sec:def-menu}
\label{sec:def-menu-descr}

The starting point of our framework for changing mechanism descriptions is the following characterization of strategyproofness in terms of applicants' \emph{menus}.\footnote{
  \autoref{def:menu} has been considered under many different names in many different contexts (e.g., \emph{sets that decentralize the mechanism} in \cite{Hammond79}; \emph{option sets} in \cite{BarberaSZ91}; \emph{proper budget sets} in \cite{LeshnoL21}; \emph{feasible sets} in \cite{KatuscakK2020}; and likely others).
  We follow the ``economics and computation'' literature (e.g., \citealp{HartN17,Dobzinski16b,BabaioffGN17}) in calling these sets ``menus.''
  This notion is distinct from many other definitions of menus \citep[e.g., those of][and many others]{MackenzieZ20, BoH20}.
}

\begin{definition}[Menu]
  \label{def:menu}
  For any matching mechanism $f$, applicant $i$, and $\succ_{-i}\in\T_{-i}$, the \emph{menu} $\Menu_{\succ_{-i}}$ of $i$ with respect to $\succ_{-i}$ is the set of all institutions $h$ for which there exists some $\succ_i \in \T_i$
  such that $f_i(\succ_i, \succ_{-i}) = h$.
  That is,
  \[ \Menu_{\succ_{-i}} = 
    \left\{\ f_i(\succ_i, \succ_{-i}) \ \middle|\ \succ_i \in \T_i\  \right\} . \]
\end{definition}

\begin{theorem}[\citealp{Hammond79}]
  \label{thrm:TaxationPrinciple}
  A matching mechanism $f$ is strategyproof if and only if
  each applicant $i$ always receives her favorite
  institution from her menu.
  That is, for every $\succ_{-i}\in\T_{-i}$ and
  $\succ_i \in \T_i$, it holds that $f_i(\succ_i, \succ_{-i}) \succeq_i h$ for all $h \in
  \Menu_{\succ_{-i}}$.
\end{theorem}
\begin{proof}
  Suppose $f$ is strategyproof and fix $\succ_{-i} \in \T_{-i}$.
  For every $\succ_i \in \T_{i}$, it holds by definition that for every $h = f_i(\succ_i',\succ_{-i}) \in \Menu_{\succ_{-i}}$, we have $f_i(\succ_i,\succ_{-i}) \succeq_i h$. 
  On the other hand, if applicant $i$ always receives her favorite institution from her menu, then she always prefers reporting $\succ_i$ at least as much as any $\succ_i'$, so $f$ is strategyproof.
\end{proof}

We use menus to describe mechanisms while exposing their strategyproofness:
\begin{definition}[Menu Description]
A \emph{menu description} of mechanism $f$ for applicant $i$ is a description %
with the following outline: 
\begin{enumerate}[]
    \item \textbf{Step\,(1)} uses only $\succ_{-i}\in\T_{-i}$ to calculate the menu $\Menu_{\succ_{-i}}$ of applicant $i$. %
    \item \textbf{Step\,(2)} uses $\succ_i\in\T_i$ to match applicant $i$ to her favorite institution in $\Menu_{\succ_{-i}}$.
\end{enumerate}
Formally, a menu description for $i$ is thus an algorithm that initially receives only $\succ_{-i}$ as input and calculates $\Menu_{\succ_{-i}}$ as an intermediate result, then additionally receives $\succ_i$ as input and uses it to calculate $i$'s favorite choice from $\Menu_{\succ_{-i}}$ as the final result.
\end{definition}

Note that while Step (1)---the calculation of $\Menu_{\succ_{-i}}$---varies between different mechanisms and menu descriptions, Step (2)---$i$'s choice from $\Menu_{\succ_{-i}}$---is essentially the same across all menu descriptions.

The central premise of our paper is that menu descriptions are one way to expose strategyproofness.
This is because the strategyproofness of any menu description can immediately be seen via a simple, one-sentence argument: First, applicant $i$'s report cannot affect her menu, and second, straightforward reporting (``truthtelling'') gets applicant $i$ her favorite institution from the menu.\footnote{
  \label{footnote:unique}
  There is also a precise sense in which menu descriptions are the \emph{only} ones for which the above argument for strategyproofness goes through. 
  In particular, suppose a description calculates the match of applicant $i$ in some mechanism $f$, and has the following outline:
    \begin{enumerate}[]
      \itemsep0em
      \item \textbf{Step\,(1)} uses $\succ_{-i}$ to calculate a set $S$ of institutions.
      \item \textbf{Step\,(2)} uses $\succ_i$ to match $i$ to her top-ranked institution in $S$.
    \end{enumerate}
  Then, it is not hard to show that the set $S$ must be $i$'s menu.
}

\subsection{Menu-in-Outcome Descriptions}
\label{sec:prelims-individ-dict}

Beyond menu descriptions, outcome descriptions that contain menu descriptions also play a key role in our results. 
We call these \emph{menu-in-outcome descriptions}.

\begin{definition}[Menu-in-Outcome Description]
A \emph{menu-in-outcome description} of mechanism $f$ for applicant $i$ is an outcome description of $f$ that contains a menu description for applicant $i$.
Equivalently, it is a description with the following outline:
\begin{enumerate}[]
\label{en:individ-dict-definition}
    \item \textbf{Step\,(1)} uses only $\succ_{-i}\in\T_{-i}$ to calculate the menu $\Menu_{\succ_{-i}}$ of applicant $i$. %
    \item \textbf{Step\,(2)} uses $\succ_i\in\T_i$ to match applicant $i$ to her favorite institution from $\Menu_{\succ_{-i}}$.
    \item \textbf{Step\,(3)} uses both $\succ_i$ and $\succ_{-i}$ to calculate the full outcome matching $f(\succ_i,\succ_{-i})$.
\end{enumerate}
Formally, a menu-in-outcome description for $i$ is thus an algorithm that initially receives $\succ_{-i}$ as input and calculates $\Menu_{\succ_{-i}}$, then additionally receives $\succ_i$ as input and calculates $i$'s favorite choice from $\Menu_{\succ_{-i}}$, and finally proceeds to calculate the entire outcome matching $f(\succ_i, \succ_{-i})$ as the final result.
\end{definition}

For example, consider SD (\autoref{def:SD}).
This mechanism is easily seen to be strategyproof, directly from its traditional description (and even for many students encountering it for the first time).
This is reflected by the fact that applicants are matched in SD via menu descriptions.
In particular, when applicants are prioritized $1 \succ 2 \succ \ldots \succ n$, the traditional description 
of SD can be divided into three steps for any applicant $i$: %
\begin{enumerate}[(1)]
  \itemsep0em
  \item Each applicant $j < i$ is matched, in order, to her top-ranked remaining institution.
  \item Applicant $i$ is matched to her top-ranked remaining institution.
  \item Each applicant $j > i$ is matched, in order, to her top-ranked remaining institution.
\end{enumerate}
Steps~(1) and~(2) form a menu description, but this menu description is contained within the (traditional) outcome description, and thus Steps~(1) through~(3) form a menu-in-outcome description.

\subsection{Uses of Menu Descriptions}
\label{sec:prelims-brute-force-menu}

The positive results of our paper present new menu descriptions of DA and TTC.
Before giving these results, we note that \emph{every} strategyproof mechanism has a menu description, given by an argument from \citet{Hammond79}.\footnote{
    This menu description was also identified by \citet{KatuscakK2020}.
}
\begin{example}[A ``brute force'' menu description]
  \label{ex:generic-mtr-matching}
  Consider any strategyproof matching mechanism $f$, and let $D$ be an outcome description of $f$.
  For each institution $h$, let $\{h\}$ denote the preference list that ranks only $h$ as acceptable.
  Then, consider the following description for applicant $i$:
  \begin{enumerate}[(1)]
    \item %
    Start with $M=\emptyset$.
    For each institution $h$ separately, evaluate $D$ on $( \{h\}, \succ_{-i} )$; if this matches $i$ to $h$, then add $h$ to $M$.
    \item Match $i$ to her highest-ranked institution in $M$.
  \end{enumerate}
    By strategyproofness, $h$ is included in $M$ in Step~(1) if and only if $h$ is on the menu. Thus, the above provides a menu description of $f$. 
\end{example}

We explore two uses of menu descriptions:
as alternative standalone descriptions for participants,
and as a tool to help illustrate the strategyproofness of a traditional description.
A description such as \autoref{ex:generic-mtr-matching} has drawbacks in both of these uses.

First, as a standalone description, \autoref{ex:generic-mtr-matching} might---compared to the traditional descriptions presented in \autoref{sec:prelims-mechanisms}---be considered cumbersome or convoluted, since it repeats $D$ many independent times. 
Given this, we look for more-direct new menu descriptions (such as our menu description of DA in \autoref{sec:parallen-hpda}).

Second, since there is no clear relation between the outcomes of \autoref{ex:generic-mtr-matching} and those of the description $D$ (absent prior knowledge that $D$ is strategyproof), it is unclear how \autoref{ex:generic-mtr-matching} might aid in conveying the strategyproofness of $D$. 
Given this, we look for menu descriptions that are closely related to the corresponding traditional description (such as our menu-in-outcome description of TTC in \autoref{sec:ttc-individ-dict}).

\section{A Menu Description of DA}
\label{sec:parallen-hpda}

In this section, we present our first result: A novel menu description of DA.
This is \autoref{alg:Phpm} (which rephrases \autoref{fig:example-descriptions-DA} in the introduction).

\begin{algorithm}[ht]
\caption{A menu description of (\emph{applicant}-proposing) DA for applicant $i$} 
  \label{alg:Phpm}

\begin{enumerate}[(1)]
  \item Run \emph{institution}-proposing DA
      with applicant $i$ removed from the market,
      to get a matching $\mu_{-i}$.
      Let $M$ be the set of institutions $h$ such that
      $i \succ_h \mu_{-i}(h)$.
  \item Match $i$ to $i$'s highest-ranked institution in $M$.
\end{enumerate}
\end{algorithm}

Step~(1) of \autoref{alg:Phpm} begins with a modified version of the traditional description of DA (from \autoref{def:DA}).
Then, it calculates $i$'s menu as an (arguably) intuitive set of institutions: those that prefer $i$ to their match in this modified version of DA.
We speculate that many real market participants would find such a description understandable. 
(In fact, our empirical companion paper \citet{GHIT} gives evidence that many lab participants can learn this description---see \autoref{sec:related-work}.)

Crucially, \autoref{alg:Phpm} uses the \emph{institution}-proposing DA algorithm to describe DA (traditionally described via the \emph{applicant}-proposing DA algorithm).
To give intuition for why the proposing side is flipped, we show via an example that switching the proposing sides in Step~(1) of \autoref{alg:Phpm} would not suffice.
\begin{example}
Consider a market with three applicants $i,d_1,d_2$ and two institutions $h_1,h_2$. Applicants have preferences $d_1 : h_1 \succ h_2$ and $d_2 : h_2 \succ h_1$, and institutions have priorities $h_1 : d_2 \succ i \succ d_1$ and $h_2 : d_1 \succ i \succ d_2$. Running \emph{applicant}-proposing DA on these preferences without $i$ gives matching $\{(d_1,h_1), (d_2, h_2)\}$, and both $h_1$ and $h_2$ prefer $i$ to their match. However, neither $h_1$ nor $h_2$ are on $i$'s menu, since having $i$ propose to any $h_i \in \{h_1,h_2\}$ (after running applicant-proposing DA without $i$) causes a ``rejection cycle'' that results in $h_i$ rejecting $i$. 
Intuitively, institution-proposing DA fixes this issue by outputting a matching that has no potential ``applicant-proposing rejection cycles.''\footnote{
 This intuition regarding ``applicant-proposing rejection cycles'' is related to the concept of an institution-improving rotation as in \citet{GusfieldI89}.
}
Specifically, institution-proposing DA gives matching $\{(d_1,h_2), (d_2, h_1)\}$, corresponding to $i$'s menu in this example being $\emptyset$.
\end{example}

Formally, the following theorem establishes the correctness of \autoref{alg:Phpm}:
\begin{theorem}[restate=restateThrmPhpmCorrectness, name=]
  \label{thrm:phpm-correctness}
  \autoref{alg:Phpm} is a menu description of DA.
\end{theorem}

\begin{proof}
Fix institutions' priorities, an applicant $i$, and preferences $\succ_{-i}$ of applicants other than $i$.
Let $\{h\}$ denote the preference list of $i$ that reports only institution $h$ as acceptable, and let $\emptyset$ denote the preference list of $i$ that reports \emph{no} institution as acceptable.
For clarity, denote applicant-proposing DA by $\mathit{APDA}(\cdot) = \mathit{DA}(\cdot)$ and denote institution-proposing DA by $\mathit{IPDA}(\cdot)$.

Since $\mathit{APDA}$ is strategyproof %
\citep{Roth82,DubinsMachiavelliGaleShapley81}, it suffices to prove that the set of institutions calculated in Step~(1) of \autoref{alg:Phpm} is applicant $i$'s menu.
Now, for any institution $h$, we observe the following chain of equivalences:
\newcounter{savedfootnote}
\setcounter{savedfootnote}{\value{footnote}}%
\stepcounter{savedfootnote}%
\allowdisplaybreaks
\begin{align*}
& 
\text{$h$ is in the menu of $i$ in $\mathit{APDA}$ (with respect to $\succ_{-i}$)} 
\\ & \qquad \Longleftrightarrow
  \text{\big(By strategyproofness of $\mathit{APDA}$; %
  \citealp{Roth82,DubinsMachiavelliGaleShapley81}\big)} 
\\ &
\text{$i$ is matched to $h$ by $\mathit{APDA}(\{h\}, \succ_{-i})$}
\\ & \qquad \Longleftrightarrow
  \text{\big(By stability of APDA and IPDA; \citealp{GaleS62}; and}
\\ & \qquad \phantom{\Longleftrightarrow} \ \ \ \text{by the Lone Wolf / Rural Hospitals Theorem\footnotemark; %
  \citealp{RothRuralHospital86}\big)}
\\ &
\text{$i$ is matched to $h$ by $\mathit{IPDA}(\{h\}, \succ_{-i})$}
\\ & \qquad \Longleftrightarrow
  \text{\big($\mathit{IPDA}(\{h\}, \succ_{-i})$ and $\mathit{IPDA}(\emptyset, \succ_{-i})$ coincide until $h$ proposes to $i$\big)}
\\ &
\text{$h$ proposes to $i$ in $\mathit{IPDA}(\emptyset, \succ_{-i})$}
\\ & \qquad \Longleftrightarrow
\text{\big($\mathit{IPDA}(\emptyset, \succ_{-i})$ and $\mathit{IPDA}(\succ_{-i})$ produce the same matching;
}
\\ & \qquad \phantom{\Longleftrightarrow} \ \ \
  \text{in $\mathit{IPDA}$, $h$ proposes in highest-to-lowest priority order\big)}
\\ &
\text{$i$ has higher priority at $h$ than $h$'s match in $\mathit{IPDA}(\succ_{-i})$.}\qedhere
\end{align*}

\end{proof}
\setcounter{footnote}{\value{savedfootnote}}%
\footnotetext{This theorem states that the set of unmatched agents is the same in every stable matching.}

In addition to giving a arguably-appealing alternative description of DA, \autoref{thrm:phpm-correctness} provides a characterization of the menu in DA that is useful for reasoning about DA's properties.
We briefly highlight two applications.
First, one can immediately see from \autoref{alg:Phpm} that, if one applicant's priorities increase at some set of institutions, then (all other things being equal) the match of that applicant in DA can only improve \citep{BalinskiS99}.
Second, a short argument using \autoref{alg:Phpm}, which we provide in \autoref{rem:menu-corollaries}, shows that in a market with $n\!+\!1$ applicants, $n$ institutions, and uniformly random full length preference lists, applicants receive in DA roughly their $n / \log(n)$th choice in expectation---rather lower than in the case with $n$ applicants, where they receive their $\log(n)$th choice---re-proving results from \cite{AshlagiKL17, CaiT22}.

\autoref{alg:Phpm} generalizes to a broader class of stable matching markets.
In fact, in \autoref{rem:contracts-menu}, we observe that the same argument as in the above proof shows that a natural generalization of \autoref{alg:Phpm}  characterizes the menu of DA in many-to-one markets, and even in a general class of markets with contracts, namely, those considered by \cite{HatfieldM05}.

Finally, we remark that \autoref{alg:Phpm} can facilitate a proof from first-principles of the strategyproofness of (traditionally described) DA (without relying on this fact as in the proof above).
We give such a proof in \autoref{sec:missing-proofs}.
While we view this proof as theoretically appealing, and perhaps useful for classroom instruction, we believe this approach remains far too mathematically involved to convey the strategyproofness of DA's traditional description to real-world participants. 
In contrast, if a clearinghouse directly adopts \autoref{alg:Phpm} as a way to describe participants' matches in explicit detail, then strategyproofness is exposed via a simple one-sentence argument.

\section{A Menu Description of TTC that Exposes that TTC's Traditional Description is Strategyproof}
\label{sec:ttc-individ-dict}

In this section, we present our second result: A novel menu description of TTC. 
In fact, we present a menu-in-outcome description that yields a novel proof that the traditional description of TTC is strategyproof.
This is \autoref{alg:individ-dict-ttc}.

\begin{algorithm}[ht]
\caption{A menu-in-outcome description of TTC for applicant $i$} 
  \label{alg:individ-dict-ttc}

\begin{enumerate}[(1)]
  \item Using $\succ_{-i}$, iteratively eliminate as many cycles not involving applicant $i$ as possible.
  Let $M$ denote the set of remaining institutions. 
  \item Using $\succ_i$, match $i$ to her highest-ranked institution in $M$. Call this institution $h$.
  \item Using $(\succ_i, \succ_{-i})$, eliminate the cycle created when $i$ points to $h$,
  then continue to eliminate cycles until all applicants match (or exhaust their preference lists).
\end{enumerate}
\end{algorithm}

\autoref{alg:individ-dict-ttc} modifies the traditional description of TTC (only) by delaying matching applicant $i$ as long as possible.\footnote{
    This can also be thought of as running TTC, with a twist: During the first stage, applicant~$i$ does not point to any institutions. This stage lasts until no cycles exist, after which $i$ points as normal (and immediately matches to the institution she points to).} 
This accurately describes the full outcome matching since, as is well known, TTC is independent of the order in which cycles are chosen to be eliminated and matched. Formally:

\begin{theorem}
  \label{thrm:MTR-TTC-UB}
  \autoref{alg:individ-dict-ttc} is a menu-in-outcome description  of of TTC. 
\end{theorem}
\begin{proof}
  Recall from \autoref{sec:prelims-mechanisms} that it is well known that the traditional description of TTC is independent of the order in which cycles are eliminated. %
  Now, consider modifying this traditional description by delaying matching cycles involving applicant $i$ as long as possible,
  and consider the pointing graph of TTC once all remaining cycles involve applicant~$i$.
  Observe that eliminating a cycle now requires matching $i$ to her highest-ranked not-yet-matched institution; see \autoref{fig:ttc-mtr-ub} for an illustration.
  Thus, \autoref{alg:individ-dict-ttc} differs from the traditional description of TTC only in the order in which cycles are eliminated, and hence calculates the TTC outcome matching.

\begin{figure}[hbt]
  \begin{minipage}[c]{0.42\textwidth}
    \vspace{-0.1in}
  \includegraphics[width=\textwidth]{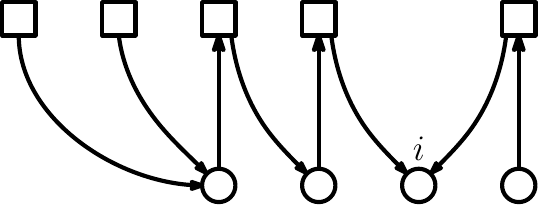}
  \end{minipage}
  \quad
  \begin{minipage}[c]{0.54\textwidth}

  \caption[TTC UB]{
    Menu calculation in \autoref{alg:individ-dict-ttc}

    \vspace{0.05in}

    {\footnotesize \textbf{Notes:}
    Circles represent applicants; squares represent institutions; each institution (resp.\ each applicant except $i$) %
    points to her favorite remaining applicant (resp.\ institution).
    Cycles not involving $i$ were already eliminated, so wherever $i$ points will form a cycle.
    \par }
  }
  \label{fig:ttc-mtr-ub}
    \end{minipage}

\end{figure}

  By construction, Step~(1) of \autoref{alg:individ-dict-ttc} does not use $\succ_i$ to calculate the set $M$.
  Thus, since $i$ is matched to her highest-ranked institution in $M$ in Step~(2), it follows that TTC is strategyproof.
  Moreover, since $i$ can match to any institution in $M$ (and only to institutions in $M$), it follows that $M$ equals $i$'s menu.\footnote{
    By the observation discussed in \autoref{footnote:unique}, the fact that $M$ equals $i$'s menu also follows from the fact that \autoref{alg:individ-dict-ttc} calculates the outcome matching of TTC as $i$'s top pick from $M$ (which is independent of $i$'s report).
  }
  Hence, \autoref{alg:individ-dict-ttc} is a menu-in-outcome description of TTC.
\end{proof}

In addition to constructing a new menu description of TTC, \autoref{thrm:MTR-TTC-UB} yields an (arguably simple and intuitive) proof that the traditional description of TTC is strategyproof.
In particular, \autoref{thrm:MTR-TTC-UB} demonstrates---given (only) the fact that TTC is independent of the order in which cycles are eliminated---that in the traditional description of TTC, any applicant $i$ is matched according to a menu description.
Hence, TTC is strategyproof.

The above simple proof is enabled by two properties of \autoref{alg:individ-dict-ttc}.
First, it contains a menu description.
Second, it only slightly tweaks TTC's traditional, outcome description (in a way that clearly maintains the same result).
Crucially, a description cannot satisfy these two properties without being an outcome description that contains a menu description, i.e., a menu-in-outcome description.

In sum, our description of TTC, and the simple argument for the strategyproofness it provides, give promising new ways to explain TTC's strategyproofness, both in the classroom and for real-world market participants.\footnote{See \autoref{footnote:morrilr}.}

\section{An Impossibility Result for DA}
\label{sec:mtr-matching}
\label{sec:applicant-linear-impossibility}

In this section, we present our third result:
an impossibility theorem showing that---in contrast to what our menu description of TTC (\autoref{sec:ttc-individ-dict}) achieves---no menu description of DA yields a simple proof of the strategyproofness of DA's traditional description. More generally, this impossibility theorem demonstrates how one might reason about such questions, which might not be clear ex ante. The proof of this theorem identifies techniques from computer science as especially useful in this context.

Concretely, we show that no slight tweak of DA's traditional description contains a menu description.
Here and throughout the paper, by ``slight tweak,'' we mean that the same result is reached, and that the step-by-step calculation is similar enough for this fact to be evident. 
We formalize this notion of slight tweaks in \autoref{sec:simple-descriptions-matching}.
We then present our impossibility theorem in \autoref{sec:DA-Impossibility}, showing that no such slight tweak of DA's traditional description contains a menu description
(and hence showing that slight tweaks cannot expose the strategyproofness of DA's traditional description in a way analogous to TTC in \autoref{sec:ttc-individ-dict}).

\subsection{Properties of Slight Tweaks of Traditional Descriptions}
\label{sec:simple-descriptions-matching}

We now identify two salient properties of the step-by-step calculations used in the traditional description of DA (and of SD, and of TTC):
\begin{itemize}
    \item First, the description only considers the preferences of each applicant once, in a specific, natural order---from favorite to least favorite.
    We call this property \emph{applicant-proposing}.
    \item Second, the description requires only a small amount of bookkeeping, namely, that required to track a single tentative matching. 
    We consider a flexible generalization of this property: that the state complexity \citep{rubinstein1986finite,abreu1988structure,oprea2020whatmakes} of the description is at most the number of possible matchings. Further generalizing, we use the property that the state complexity is at most \emph{the order of magnitude} of the number of possible matching.
    Following standard terminology from computer science, we call this property \emph{linear-memory}.
\end{itemize}

\begin{figure}[b!]
    \caption{An illustration of the traditional description of DA through an example}
      \label{fig:applicant-linear-graphics}
      \label{fig:applicant-linear-graphics:trad-DA}
    
    \begin{center}
    \centering
        \includegraphics[width=0.75\textwidth]{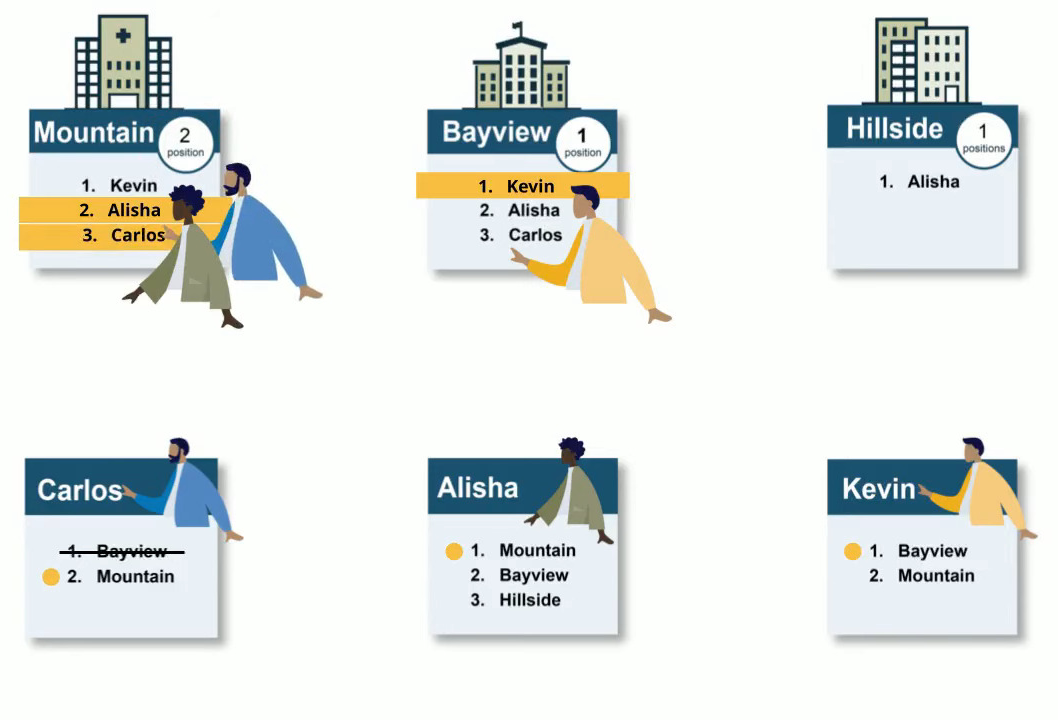}
    \end{center}
    
           \vspace{-0.3in}
      {\footnotesize \textbf{Note:} 
        Screenshot taken from \url{https://youtu.be/kVTwXNawpbk} \citep{MatchingExplained20}, a video produced by National Matching Services (the company providing matching software to the NRMP).
      \par}      
\end{figure}

Before formally defining these two properties, 
\autoref{fig:applicant-linear-graphics:trad-DA} illustrates how they are used to describe DA in one of its most celebrated practical applications: matching medical doctors to residencies in the US National Resident Matching Program (NRMP). 
The figure shows a screenshot of a video from the NRMP that relays the traditional description of DA by applying it to a small example.
The explanation in the video is aided by two visual elements: sequentially crossing off institutions from applicants' lists as the description progresses, and keeping track of a ``current tentative matching'' illustrated by the yellow-highlighted names.
In order for these two simple visual elements to illustrate the description, it is necessary that the description is applicant-proposing and linear-memory.
First, the applicant-proposing property is necessary for the video to sequentially cross off institutions from applicants' lists as the description progresses.
Second, the linear-memory property is necessary for the yellow highlighting in the video to capture all required bookkeeping.

\begin{definition}[Applicant-proposing and Linear-memory Descriptions]\hspace{-0.4em}\footnote{
  As discussed in \autoref{sec:prelims-descriptions}, we formally define descriptions to be algorithms.
  For completeness, the Supplemental \autoref{sec:model-descriptions} gives a self-contained mathematical definition of algorithms sufficient for our purposes.
}\leavevmode

  \label{def:preference-read-once}
  \label{def:preference-simple}
  \label{def:preference-linear}
  \begin{itemize}
    \item A description is \emph{applicant-proposing} if it reads applicants' preferences by querying a single applicant at a time, such that the $j$\textsuperscript{th} query to applicant $d$ returns the $j$\textsuperscript{th} institution on $d$'s preference list.
    (The priorities of the institutions, on the other hand, can be accessed by the description in any way.)
  
    Formally, an applicant-proposing description is thus a procedure that maintains some internal \emph{state} that is iteratively updated while querying applicants' preference lists (one applicant at a time), with the following property.
    For any possible inputs and for any applicant $d$, suppose the algorithm queries $d$'s preference list sequentially in states $s_1,s_2,\ldots,s_k$ as it runs, and for each $j=1,\ldots,k$, let $s_j'$ denote the updated state that the algorithm reaches immediately after querying $d$'s preferences in $s_j$.
    Then, $s_j'$ depends only on ($s_j$ and) the $j$\textsuperscript{th} institution on $d$'s preference list (which is considered to be the ``empty institution'' if $d$'s list contains fewer than $j$ institutions).

    \item The \emph{memory requirement} of a description is the number of bits required to represent the state of the description, or equivalently, the logarithm (in base 2) of the state complexity \citep{rubinstein1986finite,abreu1988structure,oprea2020whatmakes} of the description.
    Intuitively, this is the amount of bookkeeping or ``scratch paper'' required by the description.\footnote{The choice of base $2$ for the logarithm, which is equivalent to limiting to a binary alphabet on the scratch paper, is a normalization and of no consequence to our results; see \citep{nisan_segal_2006_efficient_allocations,segal_2007_social_choice_rules} for a discussion on measuring communication using a binary (rather than any other) alphabet.}

    In a matching environment with $n$ applicants and $n$ institutions, we say a description is \emph{linear-memory} if its memory requirement is  at most $\widetilde\BigO(n)$.\footnote{
        The standard computer-science notation $\widetilde\BigO(n)$ means $\BigO(n \log^\alpha n)$ for some constant
        $\alpha$. That is, for large enough $n$, memory is upper-bounded by $cn \log^\alpha n$ for some constants $c,\alpha$ that do not depend on $n$. 
        Memory $\BigO(n\log n)$ (i.e., taking $\alpha=1$) is equivalent to having state complexity that is at most any fixed polynomial in the number of possible matchings; memory $\widetilde \BigO(n)$ is slightly more permissive. 
        }
  \end{itemize}
\end{definition}

We note that linear memory is the smallest possible memory requirement for outcome descriptions (as well as for menu descriptions) of matching mechanisms.
Indeed, $\widetilde\BigO(n)$ is exactly (up to the precise logarithmic factors) the number of bits of memory required to describe a single matching (or a single applicant's menu).\footnote{
  To see this formally, note that there are $n! = 2^{O(n \log n)}$ distinct matchings involving $n$ applicants and $n$ institutions (and exactly $2^n$ possible menus). 
  Intuitively, this means that the number of letters it takes to write down a single matching with $n$ applicants and $n$ institutions (or a single menu, i.e., a subset of the $n$ institutions) is roughly proportional to $n$.
}

The applicant-proposing and linear-memory properties capture salient properties of the traditional descriptions of many matching mechanisms. As discussed above, this includes DA, but also includes SD, and TTC.\footnote{
  These properties are additionally satisfied by the popular non-strategyproof Boston mechanism
  (see, e.g., \citealp{AbdulkadirougluCYY11}).
} In particular:
\begin{itemize}
  \item In the traditional description of SD, the linear memory stores a set $S$ of not-yet-matched institutions.
  The applicant-proposing property enables the description to choose an applicant's highest-ranked institution in $S$ by reading the applicant's preference list until the first institution in $S$ is found.
  \item In the traditional description of TTC, the linear memory stores the set $S$ of not-yet-matched institutions, and a pointing graph in which some applicants point to their top-ranked institution in $S$.
  The applicant-proposing property enables the description to update an applicant's pointing edge by reading her list further down to the highest-ranked institution remaining in $S$.
\end{itemize}

Even beyond permitting these diverse traditional descriptions, the applicant-proposing and linear-memory properties
are quite flexible.
The linear-memory requirement allows for arbitrary calculations or data structures, so long as a small amount of bookkeeping per-applicant is used.
Additionally, applicant-proposing linear-memory descriptions permit many variations to the order in which applicant preferences are used by the description;
for instance, the description could query and remember one institution from \emph{each} applicant's
preference list---or could query and remember one applicant's \emph{entire} preference list.\footnote{
  We also note that if no memory requirement is considered, then \emph{every} algorithm can be implemented as an applicant-proposing description, by reading all applicants' preference lists and storing them fully in the bookkeeping, and then finally running any algorithm on these preference lists.
  See also the discussion regarding \autoref{fig:applicant-linear-graphics:memory} below.
}

Given the above, any description retaining sufficiently similar step-by-step calculations to the traditional description of DA (or SD or TTC) must, at the very least, maintain the applicant-proposing and linear-memory properties.
Slight tweaks of the traditional description of DA should retain similar step-by-step calculations, and should also calculate the same result as the traditional description, that is, be \emph{outcome descriptions}.

Overall, we thus take the view that all slight tweaks of the traditional description of DA should share three formal properties: applicant-proposing, linear-memory, and being an outcome description.\footnote{
  Note that we do not view \emph{every} description satisfying these properties as a slight tweak of a traditional one.
  Instead, we take (only) the stance that these are necessary conditions that all slight tweaks satisfy.
}
For an example for TTC, our menu-in-outcome description is a slight tweak of the traditional description; as a consequence, we have:

\begin{proposition}
\label{cor:formal-sd-ttc-id}
  \autoref{alg:individ-dict-ttc} is an applicant-proposing linear-memory menu-in-outcome descriptions of TTC. 
\end{proposition}

\subsection{Impossibility Theorem}
\label{sec:DA-Impossibility}

We now present our impossibility result. 
Using the properties discussed in \autoref{sec:simple-descriptions-matching}---applicant-proposing, linear-memory, and being an outcome description---we prove that no slight tweak of the traditional description of DA contains a menu description.

\begin{theorem}[restate=restateThrmMtrDaLb, name=]
  \label{thrm:MTR-DA-LB}
  DA has no applicant-proposing, linear-memory, menu-in-outcome description. %
  In fact, with $n$ applicants and $n$ institutions, any applicant-proposing menu-in-outcome description of DA requires $\Omega(n^2)$ memory.\footnote{
    The standard computer-science notation $\Omega(n^2)$ means that, for large enough $n$, memory is lower-bounded by $cn^2$ for some constant $c$ that does not depend on $n$.}
\end{theorem}

We prove \autoref{thrm:MTR-DA-LB} in \autoref{sec:proof-of-MTR-DA-LB} below.
The theorem shows a precise sense in which slight tweaks of DA's traditional description cannot expose its strategyproofness via a menu description. 
This is in sharp contrast to TTC, which (in the same sense) has a  slight tweak that exposes strategyproofness as shown in \autoref{sec:ttc-individ-dict}, and in contrast to SD, whose traditional description already exposes strategyproofness.

\autoref{thrm:MTR-DA-LB} is a strong impossibility result.
Namely, we show that applicant-proposing menu-in-outcome description of DA require \emph{quadratic} memory---$\Omega(n^2)$ bits.
This nearly matches the memory requirement of simply memorizing all applicants' preferences---$\widetilde\BigO(n^2)$ bits.\footnote{
  To see this formally, observe that there are $(n!)^n = 2^{O(n^2 \log(n))}$ possible preference profiles for all applicants. 
  Intuitively, this means that the number of letters it takes to write down $n$ applicants' preferences over all $n$ institutions is roughly proportional to $n^2$.
}
If an applicant-proposing description memorizes all applicants' preferences,  then it can calculate \emph{any} desired result (formally, by querying each applicant's entire preference list in order, with a separate state of the algorithm's memory for each possible preference profile, and returning a separate desired result for each such state).
This shows that quadratic memory is the highest possible amount of memory that an algorithm might require. 
Thus, where applicant-proposing menu-in-outcome description of (SD and) TTC use memory as \emph{low} as possible (linear, see \autoref{sec:simple-descriptions-matching}), for DA the memory requirement is as \emph{high} as possible (quadratic).
See \autoref{fig:applicant-linear-graphics:memory} for an illustration of the qualitative gap between these two memory requirements.

\begin{figure}[htb]
    \caption{Linear versus quadratic memory}
    \label{fig:applicant-linear-graphics:memory}
    \begin{center}
    \centering
        \includegraphics[width=0.9\textwidth]{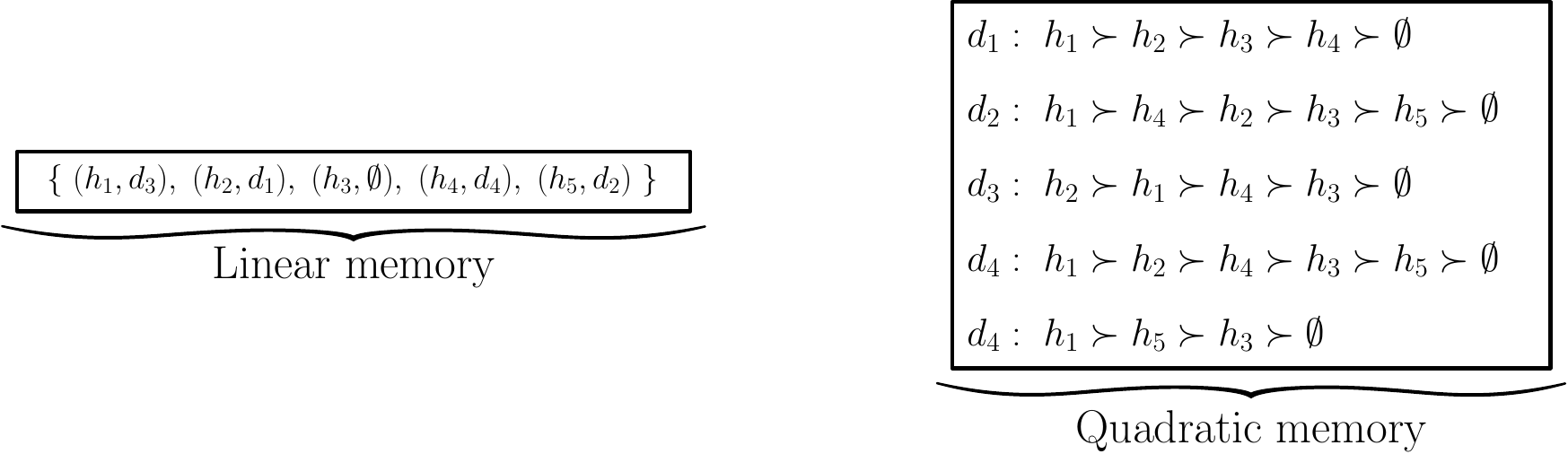}
    \end{center}
\end{figure}

\autoref{thrm:MTR-DA-LB} is also tight in the following sense.
The theorem shows that descriptions of DA cannot simultaneously satisfy four criteria: being an outcome description, containing a menu description, being applicant-proposing, and using linear-memory.
The impossibility holds only when all four of these criteria are assumed.
We establish this as follows.
First, DA's traditional description is an applicant-proposing, linear-memory outcome description.
Second, DA has an applicant-proposing \emph{quadratic memory} menu-in-outcome description, since (as discussed above) quadratic-memory is as high as possible.
Third and fourth, we show in 
the Supplemental \autoref{sec:main-body-delicate-DA-algs} 
that DA has an applicant-proposing linear-memory menu description, and a linear-memory menu-in-outcome description that is not applicant-proposing.\footnote{
    \autoref{sec:DPDA-Menu} constructs an applicant-proposing linear-memory menu description of DA, and \autoref{sec:DPDA-AKL-Alg} constructs an \emph{institution}-proposing linear-memory menu-in-outcome description of DA.
}
Hence, \autoref{thrm:MTR-DA-LB} captures the complexity of DA in our framework very precisely.

All told, our results establish a stark trichotomy---mentioned in the introduction---between SD, TTC, and DA.
The traditional description of SD is already a menu description, simultaneously for all applicants, exposing its strategyproofness easily.\footnote{
  SD has an (S)OSP implementation \citep{Li17,PyciaT21} for a similar reason.
}
The traditional description of TTC does not expose strategyproofness; however, once this description is slightly tweaked and specialized to each individual applicant, strategyproofness is exposed easily.\footnote{
  One can show that if a mechanism is not OSP-implementable---as is the case for TTC \citep{Li17}---then any description of the mechanism \emph{must} be specialized to a given applicant $i$ in order to contain a menu description for $i$.
  In \autoref{rem:ttc-order-specialized} we give a short direct proof that TTC's order requires such specialization. 
}
For DA, in contrast with both other mechanisms, no small tweak of the traditional description suffices to expose strategyproofness using a menu, in the robust and strong sense provided by \autoref{thrm:MTR-DA-LB}.

\section{Related work}
\label{sec:related-work}

Our paper is most directly inspired by the contemporary ``strategic simplicity'' program in mechanism design theory, which largely considers different dynamic implementations of mechanisms. 
A cornerstone of this literature is \citet{Li17}, which introduces obviously strategyproof (OSP) mechanisms as a way to expose strategyproofness.
Unfortunately,  TTC \citep{Li17} and DA \citep{AshlagiG18} do not have OSP mechanisms (except in rare special cases of institutions' priorities; see \citealp{Troyan19, MandalR21, Thomas21}).\footnote{
A different line of work also considers notions of strategic simplicity that are weaker than strategyproofness \citep{BorgersL19,Fernandez20,TroyanM20,Chen21,Mennle21}.
}

In contrast to the above literature, we consider different ex ante descriptions of (static, direct-revelation) mechanisms. 
\citet{BreitmoserS19} provide empirical evidence that framing a static auction as an OSP (ascending-clock) auction can be effective towards conveying strategyproofness.
Since DA and TTC do not have OSP implementations, they cannot be framed in this way. 
Nonetheless, by relaying the match of only a single applicant at a time, menu descriptions frame the mechanism in a way that is OSP for that applicant (and in fact \emph{strongly} OSP; \citealt{PyciaT21}).

The experimental paper of \citet{KatuscakK2020} also suggests describing matching mechanisms to participants via menu descriptions, but does not investigate any menu description beyond that of \autoref{ex:generic-mtr-matching}, which essentially calculates the menu by iterating over all possible reports and running the traditional mechanism description each time.

We are not aware of any prior characterizations of the menu in DA. 
Our characterization builds on a large literature developing techniques for reasoning about stable matchings.\footnote{
  In particular, our proof of \autoref{thrm:phpm-correctness} in \autoref{sec:missing-proofs} analyzes DA by incrementally modifying preference lists. 
  Similar techniques appear in \cite{GaleS85, TeoStableRevisited01, ImmorlicaM05, HatfieldM05, Gonczarowski14, AshlagiKL17, CaiT22}, for example.
  Our proof of \autoref{thrm:phpm-correctness} in \autoref{sec:parallen-hpda} uses the strategyproofness of DA; to our knowledge, this is a fairly novel technique.
  Certain other properties of DA \cite[e.g., in][]{BlumRR97,Adachi00} and of unit-demand auctions  \citep[e.g., in][]{GulS00, AlaeiKM16}, despite not being studied with relation to menus, bear some technical similarity to the menu calculation in \autoref{alg:Phpm}. However, the proofs seem unrelated.
}
The menu in DA is different than other commonly considered definitions in the theory of stable matching, such as applicant $i$'s set of stable partners \citep{GaleS62} or her budget set of institutions $h$ where she is above $h$'s cutoff \citep{Segal07, AzevedoL16, Luflade17, AzevedoB19, ImmorlicaLLL20}.
In particular, in finite matching markets, these other commonly-considered sets depend on applicant $i$'s report, and hence do not equal $i$'s menu.
We provide explicit examples and more discussion in \autoref{rem:budget-set-vs-menu} and \autoref{rem:stable-set-vs-menu}.

Proposition 2 in \citet{LeshnoL21} characterizes the menu in TTC in a different way from our \autoref{alg:individ-dict-ttc}. 
Their characterization does not give a menu-in-outcome description for TTC, and hence cannot be used in the same way as \autoref{alg:individ-dict-ttc} to derive a simple proof of the strategyproofness of TTC's traditional description.

Our paper is also loosely inspired by the literature within computer science studying menus. These works largely focus on single-player selling mechanisms 
\citep[e.g.,][]{HartN13,DaskalakisDT17,BabaioffGN17,SaxenaSW18,Gonczarowski18}.\footnote{
  \cite{BranzeiP15, GolowichL21} study the computational complexity of checking whether a mechanism, given its extensive- or normal-form representation, is strategyproof.
}
Papers considering menus in multi-player mechanisms include  \citet{Dobzinski16b} and \cite{DobzinskiRV22},
who use menus as a tool for 
bounding communication complexity.
We do not know of any prior algorithmic work on menus of matching mechanisms, nor of any prior work that analyzes different ways to {describe} multi-player mechanisms in terms of menus.

The present paper is part of our broader research agenda. 
In an earlier working paper version %
of the present article,\footnote{
For this earlier version, see \url{https://arxiv.org/abs/2209.13148v2}.
} we consider more general environments, study a basic extension of our theory for auctions, and conduct an experiment for a second-price auction and median voting.
The theoretical computer science paper \cite{GonczarowskiT24} investigates a number of complexity questions related to our results.

Most relevantly, the empirical companion paper \citet*{GHIT} investigates participants' responses to different descriptions of DA, including the traditional one and \autoref{alg:Phpm} (our menu description).
We find evidence that, while \autoref{alg:Phpm} is more complex for participants to understand than the traditional one, many participants can understand \autoref{alg:Phpm} and calculate its outcomes.
Interestingly, while levels of strategyproofness-understanding are similar under both descriptions of DA, 
we see very high levels of strategyproofness-understanding under a less-complex, stripped-down menu description that omits the details of how the menu is calculated. 
This stripped-down menu description---which relays \emph{only} strategyproofness---yields levels of strategyproofness-understanding well above a zero-information treatment benchmark, and even higher than a description relaying strategyproofness that is inspired by textbook definitions of strategyproofness.
For real-world descriptions of DA, this may suggest complementing \autoref{alg:Phpm} with a stripped-down summary focusing on the properties important for strategyproofness.\footnote{
  \citet{KatuscakK2020} show the promise of a description (in their case, the description in \autoref{ex:generic-mtr-matching} for TTC) complemented with a stripped-down summary.}

\section{Discussion}
\label{sec:conclusion}
Strategyproofness has long been proposed as a way to make mechanisms fair by leveling the playing field for players who do not strategize well \citep{PathakS08}. 
We warmly embrace this agenda. 
However, we observe that if participants do not all \emph{understand} that the mechanism is strategyproof, then disparities may remain. 
Menu descriptions may improve this understanding.
They relay ex ante how participants' matches will be calculated while ensuring that strategyproofness follows via a simple argument,
offering an alternative to status-quo tactics such as appeals to authority, asserting that the mechanism is strategyproof.\footnote{
  One common prior approach taken by clearinghouses is to encourage straightforward reporting without explaining strategyproofness.
  For example, \citet{DreyfussHR22} notes that an informative video by the National Resident Matching Program (NRMP) was formerly introduced with the text:
\begin{displayquote}
Research on the algorithm was the basis for awarding the 2012 Nobel Prize in Economic Sciences. To make the matching algorithm work best for you, create your rank order list in order of your true preferences, not how you think you will match.
\end{displayquote}\vspace{0em}
}

While menu descriptions expose strategyproofness, they may obscure other properties of the mechanism.
For example, since \autoref{alg:Phpm} (our menu description of DA) relays each applicant's match separately,
it is unclear why this description always produces a feasible (one-to-one) matching.%
\footnote{
  While traditional mechanism descriptions require participants to trust the description (as noted in, e.g., \citealp{AkbarpourL20}), 
  the fact that menu descriptions obscure feasibility may influence some participants' levels of trust. 
  While our work focuses on understanding, trust may be an interesting direction for future theoretical or empirical work.
}
In contrast, in DA's traditional description, feasibility of the outcome matching is clear, but strategyproofness is not exposed. 
\autoref{alg:individ-dict-ttc} (our menu description of TTC) might be used to simultaneously expose strategyproofness and make feasibility clear. 
Future empirical work may present TTC to lab participants using our \autoref{alg:individ-dict-ttc}---or use this description to explain the strategyproofness of TTC's traditional description (as advocated by \citealp{MorrillR24} for real-world participants)---and measure participants' understanding of both strategyproofness and feasibility.

In this paper and its experimental companion \citep*{GHIT}, 
we suggest that some principled alternative framings of mechanisms (namely, menu descriptions) might better convey their properties (namely, strategyproofness), and we analyze such framings theoretically and empirically.
We view this general premise---of reasoning about different descriptions (of the same mechanism) that expose different properties---as being of potential broader use.
Future theoretical work might consider other properties one may wish to expose (e.g., fairness or optimality) 
and study opportunities and tradeoffs for exposing these properties using different descriptions in a variety of mechanisms and settings.

\bibliography{Bib}

\clearpage 

\setcounter{page}{1}
\renewcommand{\thepage}{A.\arabic{page}}
\setcounter{table}{0} 
\renewcommand{\thetable}{A.\arabic{table}}
\setcounter{figure}{0} 
\renewcommand{\thefigure}{A.\arabic{figure}}
\setcounter{algorithm}{0} 
\renewcommand{\thealgorithm}{A.\arabic{algorithm}}
\setcounter{footnote}{0} 
\renewcommand{\thefootnote}{\arabic{footnote}}

\appendix

\part*{Appendix}

\section{Proof of the Impossibility Theorem for DA}
\label{sec:proof-of-MTR-DA-LB}

In this appendix, we prove \autoref{thrm:MTR-DA-LB}.

\autoref{thrm:MTR-DA-LB} considers applicant-proposing menu-in-outcome descriptions of DA.
Recall that such descriptions must---while reading applicants' preferences only once in favorite-to-least-favorite order---calculate $i$'s menu using $\succ_{-i}$, and then proceed to calculate the full matching using $(\succ_i, \succ_{-i})$.
The theorem states that such descriptions require quadratic memory.

To prove the theorem, we construct a large set of applicant preference profiles that, intuitively speaking, has two properties: 
(A) to calculate $i$'s menu given preferences in this set, essentially the full preference list of every applicant other than $i$ must be read in its entirety, and
(B) to calculate the final matching, essentially all this information must be remembered in its entirety.
These properties ensure that the description must store the entire preference profile in its memory.
There are many preference profiles in our construction, which implies the description has a high memory requirement.

\begin{proof}[Proof of \autoref{thrm:MTR-DA-LB}]
  Fix an applicant $i$ and let $D$ be any applicant-proposing menu-in-outcome description of DA for $i$.

  We now describe a set $\mathcal{S}\subseteq \T_{-i}$ of possible inputs to DA, illustrated in \autoref{fig:da-mtr-lb}, which allows us to establish property (B) discussed above (intuitively, by allowing $i$'s possible reports to affect the outcome matching in a different way for each different $\succ_{-i}\in\mathcal{S}$).
  For simplicity, let the number $n$ of non-$i$ applicants and institutions be a multiple of $4$.
  Other than $i$, there are applicants and institutions $d_j, d_j', h_j, h_j'$ for each $j \in \{1,\ldots,n/2\}$.
  There are $n/2$ total ``cycles'' containing two applicants  and two institutions each.
  Cycle $j$ has applicants $d_j$ and $d_j'$ and institutions $h_j$ and $h_j'$.
  The cycles are divided into two classes, ``top'' cycles (for $j\in \{1,\ldots,n/4\}$) and ``bottom'' cycles  (for $j\in\{n/4+1,\ldots,n/2\}$).

  The institutions' priorities are fixed, and defined as follows:
    \begin{align*}
    & \text{For top cycles ($j \le n/4$):} 
    &\qquad \qquad& \text{For bottom cycles ($j > n/4$):} 
    \\
    & \qquad h_j : \quad d_j' \succ i \succ d_j
    && \qquad h_j : \quad d_j' \succ d_1 \succ d_2 \succ \ldots \succ d_{n/4} \succ d_j
    \\
    & \qquad h_j' : \quad d_j \succ d_j'.
    && \qquad h_j' : \quad d_j \succ d_j'.
  \end{align*}
  For the top cycle applicants ($d_j$ with $j\le n/4$), the preferences vary (in a way we will specify momentarily). 
  Other applicants' preferences are fixed, as follows:
  \begin{align*}
    &\text{For bottom cycles ($j > n/4$):}
    && d_j : \quad h_j \succ h_j'. 
    \\
    &\text{For all cycles ($j \in \{1,\ldots,n/2\}$):}
    && d_j' :  \quad h_j' \succ h_j.
  \end{align*}

\begin{figure}[tbh]
  \begin{minipage}[c]{0.55\textwidth}
  \includegraphics[width=\textwidth]{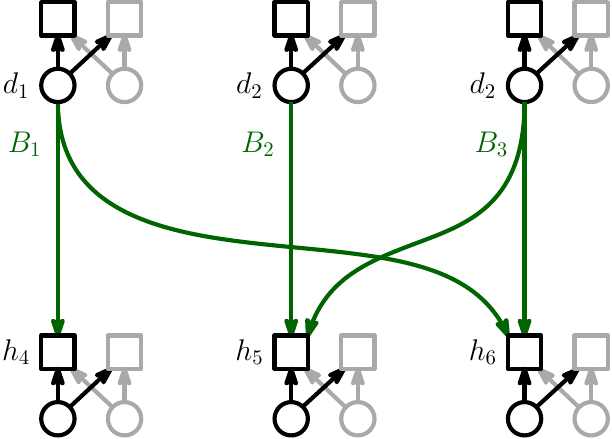}
  \end{minipage}
  \qquad
  \begin{minipage}[c]{0.38\textwidth}

  \caption[DA LB]{
    Illustration of the construction used to prove  \autoref{thrm:MTR-DA-LB}
    }
    
    {\footnotesize \textbf{Notes:}
    Dark nodes represent $d_j$ or $h_j$ for some $j$, and grey nodes represent $d_j'$ or $h_j'$.
    The green arrows directed outwards from a top cycle $d_j$ represent the sets $B_j$.
    \par
    }
  \label{fig:da-mtr-lb}
    \end{minipage}

\end{figure}

  Let $\mathcal{S}$ denote the set of preference profiles where we additionally have:
  \begin{align*}
    \text{For top cycles ($j\le n/4$):}
    &&
    d_j :  \qquad h_j \succ B_j \succ h_j',
  \end{align*}
  where $B_j$ is an arbitrary subset of $\{h_k\ |\ k > n/4 \}$, ranked in any fixed order (say, increasing order of $j$).
  Any such collection of $\big( B_j \big)_{j = 1}^{n/4}$ defines a distinct preference profile in $\mathcal{S}$. 
  Note that $|\mathcal{S}| = 2^{(n/4)^2}$.
  See \autoref{fig:da-mtr-lb} for an illustration.

  We additionally define a set of inputs $\mathcal{S}'\supseteq\mathcal{S}$, which allow us to establish property (A) discussed above (intuitively, by making $i$'s menu depend on the final institution ranked above $\emptyset$ on other applicants' lists).
  Specifically, let $\mathcal{S}'$ denote the set containing every element of $\mathcal{S}$,
  and additionally any top cycle applicant $d_j$ ($j \le n/4$) may or may not truncate the final institution $h_i'$ off her list, marking it as unacceptable.
  In other words, in addition to the sets $\big(B_j\big)_{j=1}^{n/4}$, an element of $\mathcal{S}'$ is defined by bits $(c_j)_{j=1}^{n/4}$, such that, for each top cycle $j$ ($j\le n/4$):\footnote{
    This collection of preferences can also be constructed with full preference lists by adding some unmatched institution $h_{\emptyset}$ to represent truncating $d_i$'s list.
  }
  \begin{align*}
      & \text{When $c_j = 0$:} && d_j : \quad h_j \succ B_j \succ h_j'.
      \\
      & \text{When $c_j = 1$:} && d_j : \quad h_j \succ B_j.
  \end{align*}
  
  We now proceed to prove the two crucial properties of DA, and the description $D$, when run on this family of preference profiles.
  The following lemmas formalize properties (A) and (B) discussed above, showing (respectively) that $D$ must essentially read all of the preferences in $\succ_{-i}\in\mathcal{S}'$ in order to calculate $i$'s menu, and (before knowing $\succ_i$)  must remember essentially all of this information in order to calculate the outcome matching $\text{DA}(\succ_i, \succ_{-i})$.

  \begin{lemma}
    \label{lem:dpda-mtr-lb-must-read}
    Consider a preference profile in $\mathcal{S}'$.
    For each top cycle $j$ (with $j\le n/4$), we have that $h_j$ is in applicant $i$'s menu in DA if and only if $d_j$ does not rank $h_j'$ (i.e., $c_j=1$).
    Hence, to correctly calculate $i$'s menu, description $D$ must read the entire preference list of each such $d_j$ (up to the position of $h_j'$).
  \end{lemma}
    To prove this lemma, consider the execution of the APDA algorithm when $i$ submits a list containing only $h_j$. %
    First, $d_j$ is rejected, then she proposes to every institution $h_k \in B_j$. 
    This ``rotates'' the bottom cycle containing $h_k$;
    in more detail, $h_k$ will accept the proposal from $d_j$, then $d_k$ will propose to $h_k'$, then $d_k'$ with propose to $h_k$, and $d_j$ will be rejected from $h_k$.
    This will occur for every $h_k \in B_j$, so $d_j$ will not match to any $h_k$ with $k\in\{n/4+1,\ldots,n/2\}$. 

    Finally, after getting rejected from each institution in $B_j$, applicant $d_j$ may or may not propose to $h_j'$, depending on the bit $c_j$. 
    If she does not, then $d_*$ remains matched to $h_j$ and in this case $h_j$ is on $i$'s menu. 
    If she does, then $h_j'$ will reject $d_j'$, who will propose to $h_j$, which will reject $i$.
    So $i$ will go unmatched, and thus in this case $h_j$ is not on $i$'s menu.
    
    The final sentence of the lemma then follows from the fact that $D$ is an applicant-proposing and must calculate $i$'s menu. %
    This proves \autoref{lem:dpda-mtr-lb-must-read}.

  \begin{lemma}
    \label{claim:da-mtr-lb-with-gtc}
    Each distinct preference profile $\succ_{-i} \in \mathcal{S}$ induces a distinct function $\text{DA}(\cdot, \succ_{-i}) : \T_{i} \to A$ from applicant $i$'s report to outcome matchings.
    Hence, to correctly calculate the outcome matching, the description $D$ must---across all states where it finishes calculating $i$'s menu---have at least one state for each element of $\mathcal{S}$.
  \end{lemma}
    To prove this lemma, consider two distinct preference profiles in $\mathcal{S}$, one
    profile $\succ_{-i}$ corresponding to $\big( B_j \big)_{j=1}^{n/4}$, and the other
    profile $\succ_{-i}'$ corresponding to $\big( B_j' \big)_{j=1}^{n/4}$.
    Without loss of generality, there is some $j$ and $k$ such that $h_k \in B_j \setminus B_j'$.
    Suppose now that $i$'s report $\succ_i$ lists only $h_j$.
    Then, consider execution of the APDA algorithm under $(\succ_i, \succ_{-i})$ and under $(\succ_i, \succ_{-i}')$.
    Under $\succ_{-i}$, the bottom tier cycle containing $h_k$ will be
    ``rotated,'' i.e. since $h_k \in B_j$, the sequence of rejections will cause $h_k$ to match to $d_k'$.
    However, this is not the case under $\succ_{-i}'$, since $h_k \notin B_j$.
    Thus, $\text{DA}(\cdot,\succ_{-i}) \ne \text{DA}(\cdot,\succ_{-i}')$.

    We now prove the second sentence of the lemma.
    As argued in \autoref{lem:dpda-mtr-lb-must-read}, $D$ must have read all top cycle applicants' preferences in order to calculate $i$'s menu.
    Moreover, since $D$ is a menu-in-outcome description, it must do so before learning $\succ_i$. Hence, to calculate the outcome matching correctly at the end, $D$ must %
    remember the entirety of
    $\big( B_j \big)_{j=1}^{n/4}$.
    This proves \autoref{claim:da-mtr-lb-with-gtc}.

  We now prove \autoref{thrm:MTR-DA-LB}.
  Together, \autoref{lem:dpda-mtr-lb-must-read} and \autoref{claim:da-mtr-lb-with-gtc} show that when $D$ has just calculated the menu of applicant $i$, the description must be in a distinct state for each distinct $\succ_{-i} \in \mathcal{S}$. 
  There is one such $\succ_{-i}$ for each possible way of assigning the sets $B_j\subseteq %
  \{ h_k \ |\ k > n/4 \}$ for all $j\in\{1,\ldots,n/4\}$. 
  For each such $j$, there are $2^{n/4}$ ways to assign $B_j$, and hence
  there are $\bigl(2^{n/4}\bigr)^{n/4} = 2^{(n/4)^2} = 2^{\Omega(n^2)}$ possible ways to set this collection $\big(B_j\big)_{j=1}^{n/4}$. 
  Thus, %
  the description requires at least this many states, and thus requires memory $\Omega(n^2)$.
  This finishes the proof.
\end{proof}

\section{Additional Proofs and Remarks}
\label{sec:missing-proofs}

In this appendix, we provide additional supplemental proofs and remarks omitted from the main text.

We start by proving \autoref{thrm:phpm-correctness}, which shows that \autoref{alg:Phpm} is a menu description of DA, without assuming the strategyproofness of DA.
This provide an alternative, potentially-instructive approach for proving DA's strategyproofness.

\begin{proof}[Alternative proof of \autoref{thrm:phpm-correctness}]
  
  We show that, for any applicant $i$, \autoref{alg:Phpm} correctly calculates $i$'s match in DA. 
  To this end, fix the priorities of institutions and preferences $\succ = (\succ_i, \succ_{-i})$ of all applicants.
  Let $h_* = \mathit{\mathit{APDA}}_{i}(\succ)$ denote the match of $i$ according to applicant-proposing DA. 
  Our goal is to show that $h_*$ is the outcome of \autoref{alg:Phpm}.
  Hence, we must show $h_*$ is the $\succ_i$-favorite institution in the set containing 
  (1) the ``outside option'' of going unmatched, and 
  (2) all institutions $h$ such that $h$ prefers $i$ to $\mathit{IPDA}_h(\succ_{-i})$ (the match of $h$ according to institution-proposing DA in the market without $i$).
  
  Let $\emptyset$ denote the empty preference report of $i$ (i.e., the report marking all institutions as unacceptable).
  Observe that $\mathit{IPDA}(\succ_{-i})$ and $\mathit{IPDA}(\emptyset, \succ_{-i})$ match applicants (other than $i$) in exactly the same way, %
  and furthermore, the institutions $h$ that prefer $i$ to $\mathit{IPDA}_h(\succ_{-i})$ are exactly those that propose to $i$ during the calculation of $\mathit{IPDA}(\emptyset,\allowbreak \succ_{-i})$. 
  Therefore, it suffices to prove:
  \begin{enumerate}
      \item[(I)] If $h_*\ne\emptyset$, then then $h_*$ proposes to $i$ during $\mathit{IPDA}(\emptyset, \succ_{-i})$.
      \item[(II)] Applicant $i$ gets no proposal in $\mathit{IPDA}(\emptyset, \succ_{-i})$ that is $\succ_i$-preferred to $h_*$.
  \end{enumerate}

  We start by proving (I). Assume that $h_*\ne\emptyset$. Let $\{h_*\}$ denote the preference list of $i$ ranking only $h_*$ (i.e., marking all other institutions as unacceptable). Observe that $\mathit{APDA}(\succ_i,\succ_{-i})$ is also stable under preferences $(\{h_*\},\succ_{-i})$. 
  Thus, by the Lone Wolf / Rural Hospitals Theorem 
  \citep{RothRuralHospital86} (and by stability of APDA and IPDA; \citealp{GaleS62})
  since $i$ is matched in $\mathit{APDA}(\succ_i,\succ_{-i})$, she must be matched in $\mathit{IPDA}(\{h_*\}, \succ_{-i})$ as well.
  Thus, $\mathit{IPDA}(\{h_*\}, \succ_{-i})=h_*$. 
  Since $\mathit{IPDA}(\{h_*\},\succ_{-i})$ and $\mathit{IPDA}(\emptyset,\succ_{-i})$ coincide until $h_*$ proposes to $i$, we conclude that $h_*$ must propose to $i$ during $\mathit{IPDA}(\emptyset, \succ_{-i})$, proving (I).
  
  We now prove (II).
  Let $T$ denote $i$'s preference list, truncated at and below $h^*$, i.e., the report listing only institutions that $i$ strictly prefers to $h^*$.
  Observe that $i$ must go unmatched in $\mathit{APDA}(T,\succ_{-i})$, since every proposal by $i$ before $h_*$ was rejected in $\mathit{APDA}(\succ_i,\succ_{-i})$.
  Hence, by the Lone Wolf / Rural Hospitals Theorem  %
  \citep{RothRuralHospital86}, $i$~goes unmatched in $\mathit{IPDA}(T, \succ_{-i})$.
  Now, since $i$ goes unmatched in $\mathit{IPDA}(T, \succ_{-i})$, we see that $i$ does not receive any proposal that is $\succ_i$-preferred to $h_*$ in $\mathit{IPDA}(\emptyset,\succ_{-i})$, proving (II).

  We have shown that the outcome calculated at the end of Step~(2) of \autoref{alg:Phpm} is $i$'s outcome in DA.
  Moreover, observe that the set calculated in Step~(1) of is independent of $i$'s report.
  Hence, DA is strategyproof (by the same proof outline that applies to every menu description).
  Moreover, as observed in \autoref{footnote:unique}, the menu is the \emph{only} set $M$ of institutions that is independent of $i$'s report such that $i$ always receives her favorite institution in $M$.
  Hence, the set in Step~(1) must be $i$'s menu, and \autoref{alg:Phpm} is a menu description of DA.
\end{proof}

\begin{remark}
  \label{rem:many-to-one-menu}
  As noted in \autoref{sec:parallen-hpda}, \autoref{thrm:phpm-correctness} extends to
  many-to-one markets with substitutable priorities. To quickly see why this extension holds in the special case in which institutions have responsive preferences (i.e., the special case in which each institution has a master preference order and a capacity), fix a many-to-one market, and following a standard approach, consider a one-to-one market where each institution from the original market is split into ``independent copies.'' That is, the number of copies of each institution equals the capacity of the institution, each ``copied'' institution has the same preference list as the original institution, and each applicant ranks all the copies of the institution (in any order) in the same way she ranked the original institution. Ignoring the artificial difference between copies of the same institution, the run of applicant-proposing DA is equivalent under these two markets. Thus, an applicant's menu is equivalent under both markets, and so by \autoref{thrm:phpm-correctness}, a menu description for the many-to-one market can be given through institution-proposing DA under the corresponding one-to-one market, which in turn is equivalent to institution-proposing DA under the original market (where at each step, each institution proposes to a number of applicants up to its capacity). The only change in \autoref{alg:Phpm} in this case would be replacing the condition $i \succ_h \mu_{-i}(h)$  with $\exists d' \in \mu_{-i}(h):i \succ_h d'$.
\end{remark}

\begin{remark}
  \label{rem:contracts-menu}
  As additionally noted in \autoref{sec:parallen-hpda}, \autoref{thrm:phpm-correctness} also extends to many-to-one markets with contracts in which the institutions have substitutable preferences that satisfy the law of aggregate demand (the conditions under which \cite{HatfieldM05} prove that the strategyproofness of applicant-proposing DA and the rural hospitals theorem hold).
  \autoref{alg:Phpm-contracts} gives a menu description of DA in this environment, which generalizes \autoref{alg:Phpm} as follows: (1) \autoref{alg:Phpm-contracts} uses the generalized Gale--Shapley algorithm of \cite{HatfieldM05} 
  starting from $(\emptyset, X)$ (where $X$ is the set of all possible contracts) to calculate the institution-optimal stable outcome without $i$ to get a matching $\mu_{-i}$. (2) A given contract $c = (i,h,c)$ (i.e., an $(\mathrm{applicant},\mathrm{institution},\mathrm{term})$ tuple) is on $i$'s menu if and only if $h$ would choose $(i,h,c)$ if given a choice from the set containing $(i,c)$ and its matches in $\mu_{-i}$ (in the notation of \cite{HatfieldM05}, $c \in C_h(\mu_{-i}(h) \cup \{c\})$).
  Under this modification, each step of the proof of \autoref{thrm:phpm-correctness} in \autoref{sec:parallen-hpda} holds by a completely analogous argument for this market.
\end{remark}

\begin{algorithm}[ht]
\caption{A menu description for applicant $i$ of the applicant-optimal stable matching in a many-to-one market with contracts} 
  \label{alg:Phpm-contracts}

\begin{enumerate}[(1)]
  \item Calculate the institution-optimal stable matching with applicant $i$ removed from the market using the generalized Gale--Shapley algorithm of \cite{HatfieldM05}. Call the resulting matching $\mu_{-i}$.
  Let $M$ be the set of contracts $c = (i,h,t)$ involving applicant $i$ such that $c \in C_h( \mu_{-i}(h) \cup \{c\} )$.
  \item Match $i$ to $i$'s highest-ranked contract in $M$.
\end{enumerate}

\end{algorithm}

\begin{remark}
  \label{rem:menu-corollaries}
  In this remark, we show how \autoref{thrm:phpm-correctness}, which characterizes the menu in DA in terms of \autoref{alg:Phpm}, can be used to prove results from \cite{AshlagiKL17} via arguments similar to \cite{CaiT22}. Consider a randomized market with $n+1$ applicants and $n$ institutions, where such that each applicant/institution draws a full-length preference list uniformly at random, and let $\mu$ be the result of (applicant-optimal) DA with these preferences. We prove that the expected rank each applicant receives on their preference list (formally, the expectation of $|\{h : h \succeq_d \mu(d)\}|$ for any $d$) is at least $(1 - \epsilon) n / \log(n)$ for any $\epsilon > 0$ and large enough $n$.

  Fix an applicant $d_*$, and consider calculating $d_*$'s menu using \autoref{alg:Phpm} in this market. This is equivalent to considering IPDA in a market where $d_*$ rejects all proposals, and setting $d_*$'s menu to consist of all proposals she receives. 
  By the principle of deferred decisions, this run of IPDA can be constructed by letting each institution $h$ proposes to a uniformly random applicant (among those $h$ has not yet proposed to) each time she proposes. 
  Observe that this run of IPDA will terminate as soon as each of the $n$ applicants other than $d_*$ receives a proposal. 
  Thus (much like the standard case of $n$ applicants and $n$ institutions in APDA \cite{WilsonAnalysisStable72}), the total number of proposals made in this run of IPDA is stochastically dominated by a coupon collector random variable. Thus, intuitively, the total number of proposals will be $n \log(n)$, and $\log(n)$ of these will go to $d_*$ in expectation, and $d_*$'s top choice out of these $\log(n)$ proposals will be their $n / \log(n)$th ranked choice overall.
  
  Formally, let $Y$ denote the number of proposals $d_*$ receives, and let $\overline{Y}$ denote the same quantity in a market where each institution makes each proposal completely uniformly at random (without regard to prior proposals); it follows that $Y$ is stochastically dominated by $\overline{Y}$. 
  Let $\overline{Z_i}$ denote the total number of proposals between the $(i-1)$th and $i$th distinct applicant in $\D\setminus\{d_*\}$ receiving a proposal (in the market with repeated proposals). The expected value of $Z_i$ is exactly $({n+1})/({n+1-i})$, and each of these $Z_i$ proposals (except for the final one) has a $1/i$ probability of going to $d_*$. Thus, we have
  \[ \E{{Y}} 
    \le \E{\overline{Y}} 
    = \sum_{i=1}^n \frac{1}{i}\left(\frac{n+1}{n+1-i} - 1\right)
    = \sum_{i=1}^n \frac{1}{i}\left(\frac{i}{n+1-i}\right)
    = H_n
    \le \log(n) + 1.
  \] 
  Now, let $R = |\{h : h \succeq_d h_*\}|$, where $h_*$ is $d_*$'s top-ranked proposal received (i.e., $d_*$'s match in APDA). One can show that, conditioned on $Y = y$, we have the expected value of $R$ exactly equal to $(n+1)/(y+1)$ (see for example  \cite[Claim A.1]{CaiT22}). Thus, by Jensen's inequality, we have
  \[ \E{R} = \Es{y \sim Y}{\frac{n+1}{y+1}}
  \ge \frac{n+1}{\E{Y}+1} 
  \ge \frac{n+1}{\log(n)+2}
  \ge \big(1 - \epsilon\big)\frac{n}{\log(n)}
  \]
  for any $\epsilon > 0$ and large enough $n$, as desired.
\end{remark}

\begin{remark} \label{rem:ttc-order-specialized}
  We now formally show that, unlike SD, a description of TTC \emph{must} be specialized to individual applicants in order to contain a menu description for them.
  Concretely, we show that any outcome description of TTC cannot contain a menu description for \emph{two} applicants (where, in contrast, our \autoref{alg:individ-dict-ttc} contains a menu description for exactly one student).
  
  To do this, it suffices to construct an instance containing two applicants $d_1$ and $d_2$ such that each of their menus depends on the other. 
  For example, consider an instance where $h_i : d_i \succ d_{3-i}$ and $d_i : h_{3-i} \succ h_i$ for $i\in\{1,2\}$.
  Under this instance, for each $i\in\{1,2\}$, institution $h_{3-i}$ is on $d_i$'s menu, but if applicant $d_{3-i}$ changed her preference list, this would no longer be true.
  Hence, a description cannot calculate either applicant's menu before the description queries the other applicant's type.
\end{remark}

\begin{remark}
\label{rem:budget-set-vs-menu}
    We now show that in (finite-market) DA, budget sets and menus are different sets; moreover, we show that neither set includes the other.
    For a fixed profile of preferences and priorities, denote an applicant $i$'s budget set $B(i) = \{ h | i \succeq_h \mu(i) \}$, where $\mu$ is the outcome of DA.
    Let $M(i) = \Menu_{\succ_{-i}}$ denote $i$'s menu.

    Now, consider the market with institutions $h_1, h_2, h_3,$ and $h_4$, and applicants $d_1, d_2, d_3,$ and $d_4$.
    Let the preferences and priorities be as follows:
    \begin{align*}
        & h_1: d_1 \succ d_2 \succ \ldots
        && d_1: h_1 \succ \ldots
        \\
        & h_2: d_4 \succ d_3 \succ d_2 \succ d_1 \succ \ldots 
        && d_2: h_1 \succ h_2 \succ h_4 \succ \ldots
        \\
        & h_3: d_3 \succ \ldots 
        && d_3: h_3 \succ \ldots
        \\
        & h_4: d_2 \succ d_4 \succ \ldots
        && d_4: h_4 \succ h_2 \succ \ldots
    \end{align*}
  Then, one can check that DA pairs $h_i$ to $d_i$ for each $i=1,\ldots,4$,
  and that $h_2 \in B(d_3) \setminus M(d_3)$, and also $h_2 \in M(d_1) \setminus B(d_1)$.
  Thus, neither the menu nor the budget set contain the other.
  Moreover, the relationship between the two sets does not seem to be restricted in a straightforward way based on priorities and the outcome of DA:
  despite the fact that $d_3\succ_{h_2} d_2$, we have $h_2\notin M(d_3)$;
  despite $d_1\prec_{h_2} d_2$, we have $h_2 \in M(d_1)$.
\end{remark}

\begin{remark}
\label{rem:stable-set-vs-menu}
    We now show that in DA, an applicant's set of stable partners is a (possibly strict) subset of her menu.
    For a given profile of preferences and priorities, let $S(i)$ denote the set of stable partners of applicant $i$, and let $M(i)$ denote her menu.
    We begin by showing that $M(i)\ne S(i)$.
    Consider any instance with two institutions $h_1, h_2$ which both rank $i$ above all other applicants.
    Both $h_1$ and $h_2$ must be in $i$'s menu.
    However, if $i$ ranks $h_1$ above all other institutions, then $h_1$ is $i$'s unique stable partner;
    thus $h_2 \in M(i)\setminus S(i)$.

    We now show that $S(i)\subseteq M(i)$. 
    Suppose the profile of preferences and priorities is $P$.
    Consider any $h \in S(i)$, and let $\mu$ be a stable matching with $\mu(i)=h$. 
    Then, let $\widetilde P$ denote modifying $P$ by having $i$ submit a list which ranks only $h$. 
    Then, observe that $\mu$ is also stable under $\widetilde P$.
    Thus, by the Rural Hospital Theorem %
    \citep{RothRuralHospital86}, $i$ and $h$ must be matched in every stable matching under $\widetilde P$, in particular, in $\mathit{DA}(\widetilde P)$. 
    Thus, $h \in M(i)$, and $S(i)\subseteq M(i)$.
\end{remark}

\clearpage 

\setcounter{page}{1}
\renewcommand{\thepage}{S.\arabic{page}}
\setcounter{table}{0} 
\renewcommand{\thetable}{S.\arabic{table}}
\setcounter{figure}{0} 
\renewcommand{\thefigure}{S.\arabic{figure}}
\setcounter{algorithm}{0} 
\renewcommand{\thealgorithm}{S.\arabic{algorithm}}
\setcounter{footnote}{0} 
\renewcommand{\thefootnote}{\arabic{footnote}}

\part*{Supplemental Material}

\renewcommand{\thesection}{S}

\section{Mathematical Model of Algorithms}
\label{sec:model-descriptions}

In this appendix, we define from first-principles a mathematical model of descriptions of mechanisms which can express all our results.

We introduce the notion of an \emph{extensive-form description}.
For generality, we state this definition in terms of a general mechanism design environment with players $1,\ldots,n$, type spaces $\T_1,\ldots,\T_n$, and outcome space $A$.
At a technical level, an extensive-form description is similar to an extensive-form mechanism, except that different branches may ``merge,''
i.e., %
the underlying game tree is actually a directed acyclic graph (DAG).\footnote{
  Alternatively, extensive-form descriptions can be viewed as finite automata where state transitions are given by querying the types of players. %
}
Note, however, that the interpretation is different from that of an extensive-form mechanism: Rather than modeling an interactive process where the players may act multiple times, an extensive-form description spells out the steps used to calculate some result by iteratively querying the directly-reported types of the players.

We formally define three types of extensive-form descriptions, corresponding to our three description outlines: outcome descriptions, menu descriptions, and menu-in-outcome descriptions.

\begin{definition}[Extensive-Form Descriptions]
  \label{def:decision-dag}\leavevmode
  \begin{itemize}
  \item
  An \emph{extensive-form description} in some environment is defined by a directed graph on some set of vertices $V$.\footnote{
    Formally, a directed graph $G$ on vertices $V$ is some set of ordered pairs $G\subseteq V\times V$. An element $(v,w)\in G$ is called an \emph{edge} from $v$ to $w$. A \emph{source} (resp., \emph{sink}) vertex is any $v$ where there exists no vertex $w$ with an edge from $w$ to $v$ (resp., from $v$ to $w$).
  } 
  There is a (single) root vertex $s \in V$, and the vertices of $V$ are organized into \emph{layers} $j=1,\ldots,L$ such that each edge goes between layer $j$ and $j+1$ for some $j$. For a vertex $v$, let $S(v)$ denote the edges outgoing from~$v$. Each vertex $v$ with out-degree at least~$2$ is associated with some player $i$, whom the vertex is said to \emph{query}, and some \emph{transition function} $\ell_v: \T_i \to S(v)$ from types of player $i$ to edges outgoing from $v$. (It will be convenient to also allow vertices with out-degree~$1$, which are not associated with any player.) For each type profile $(t_1,\ldots,t_n)$, the \emph{evaluation path} on $(t_1,\ldots,t_n) \in \T_1\times\ldots\times\T_n$ is defined as follows: Start in the root vertex~$s$, and whenever reaching any non-terminal vertex $v$ that queries a player $i$ and has transition function $\ell_v$, follow the edge $\ell_v(t_i)$.

  \item
  An \emph{extensive-form outcome description} of a mechanism $f$ is an extensive-form description in which each terminal vertex is labeled by an outcome, such that for each type profile $(t_1,\ldots,t_n)\in\T_1\times\T_n$, the terminal vertex reached by following the evaluation path on $t\in T$ is labeled by the outcome $f(t_1,\ldots,t_n)$.

  \item
  An \emph{extensive-form menu description} of a social choice function $f$ for player $i$ is an extensive-form description with $k+1$ layers, such that 
  (a) each vertex preceding layer~$k$ queries some player other than $i$,
  (b) each vertex $v$ in layer $k$ queries player $i$ and is labeled by some set $M(v)\subseteq A_i$, such that if $v$ is on the evaluation path on a type profile $(t_1,\ldots,t_n)\in\T_1\times\T_n$, then $M(v) = \Menu_{t_{-i}}$ is the menu of player~$i$ with respect to $t_{-i}$ in $f$, and 
  (c) each (terminal) vertex $v$ in the final layer $k+1$ is labeled by an outcome for player $i$,\footnote{
    Formally, in a general mechanism design environment, an \emph{outcome of player $i$} (or, an \emph{$i$-outcome}) is a maximal set $E$ of outcomes such that all possible types of player $i$ in $\T_i$ view each outcome in $E$ as equally desirable. 
  } such that if $v$ is reached by following the evaluation path on a type profile $(t_1,\ldots,t_n)$, then $v$ is labeled by $i$'s outcome in $f(t_1,\ldots,t_n)$.

  \item
  \label{def:MTR}
  An \emph{extensive-form menu-in-outcome description} of $f$ for player $i$ is an extensive-form outcome description such that, for some $k$, the first $k+1$ layers are an extensive-form menu description.
  \end{itemize}
\end{definition}

For a concrete example of an extensive-form description, we consider a menu description of a second price auction.\footnote{
  While we have not formally defined menus or menu descriptions in non-matching environments, they naturally generalize by considering the menu of $i$ induced by reports $t_{-i}$ to be the set of $i$'s outcomes consistent with $t_{-i}$. 
} 
In this mechanism, a bidder's menu consists of two options: winning the item and playing the highest bid placed by any other bidder, or winning nothing and paying nothing.
Thus, a menu description can be given as follows: 
\begin{enumerate}[(1)]
    \item Your ``price to win'' the item will be set to the highest bid placed by any other player.
    \item  If your bid is higher than this ``price to win,'' then you will win the item and pay this price. Otherwise, you will win nothing and pay nothing.
\end{enumerate}

An extensive-form description can formalize this menu description by querying the other bidders one-by-one, while keeping track of only the highest bid placed by any of them. \autoref{fig:full-menu-only-2pa} provides an illustration.

\begin{figure}[hbt]

\tikzset{%
  zeroarrow/.style = {-stealth,dashed},
  onearrow/.style = {-stealth,solid},
  c/.style = {circle,draw,solid,minimum width=1em,
        minimum height=1em},
  r/.style = {rectangle,draw,solid,minimum width=2em,
        minimum height=2em}
}   
\begin{minipage}{0.55\textwidth}
  \hspace{-0.2in}
\begin{tikzpicture}[node distance=1cm and 1cm]\scriptsize
   \node[c] (emp) {};

   \node[c] (l13) at (1.4,1.5)  {\$3};
   \node[c] (l12) at (1.4,0.5)  {\$2};
   \node[c] (l11) at (1.4,-0.5) {\$1};
   \node[c] (l10) at (1.4,-1.5) {\$0};

   \draw[onearrow] (emp) -- (l13) node[near start, above, yshift=5pt]{\$3};
   \draw[onearrow] (emp) -- (l12) node[near end, above]    {\$2};
   \draw[onearrow] (emp) -- (l11) node[near end, above]    {\$1};
   \draw[onearrow] (emp) -- (l10) node[near start, below, yshift=-3pt]{\$0};

   \node[c] (l23) at (3.0,1.5)  {\$3};
   \node[c] (l22) at (3.0,0.5)  {\$2};
   \node[c] (l21) at (3.0,-0.5) {\$1};
   \node[c] (l20) at (3.0,-1.5) {\$0};

   \draw[onearrow] (l10) -- (l20) ;
   \draw[onearrow] (l10) -- (l21) ;
   \draw[onearrow] (l10) -- (l22) ;
   \draw[onearrow] (l10) -- (l23) ;
   \draw[onearrow] (l11) -- (l21) ;
   \draw[onearrow] (l11) -- (l22) ;
   \draw[onearrow] (l11) -- (l23) ;
   \draw[onearrow] (l12) -- (l22) ;
   \draw[onearrow] (l12) -- (l23) ;
   \draw[onearrow] (l13) -- (l23) ;

   \node (dots) at (3.7,0.0) {{\large\ldots}};

   \node[c] (l33) at (5.0,1.5)  {\$3};
   \node[c] (l32) at (5.0,0.5)  {\$2};
   \node[c] (l31) at (5.0,-0.5) {\$1};
   \node[c] (l30) at (5.0,-1.5) {\$0};

   \node (p1) at (0.1,  -1.2) {{\normalsize$\underbrace{\ \ \ \ }_{\text{Bidder 1}}$}};
   \node (p2) at (1.4,-2.1) {{\normalsize$\underbrace{\ \ \ \ }_{\text{Bidder 2}}$}};
   \node (p3) at (3,  -2.1) {{\normalsize$\underbrace{\ \ \ \ }_{\text{Bidder 3}}$}};
   \node (pn) at (5,  -2.1) 
     {{\normalsize$\underbrace{\ \ \ \ }_{\text{Bidder } n}$}};

   \draw[onearrow] ($($(l20)!.5!(l30)$)$) -- (l30) ;
   \draw[onearrow] ($($(l20)!.55!(l31)$)$) -- (l31) ;
   \draw[onearrow] ($($(l20)!.7!(l32)$)$) -- (l32) ;
   \draw[onearrow] ($($(l20)!.7!(l33)$)$) -- (l33) ;
   \draw[onearrow] ($($(l21)!.5!(l31)$)$) -- (l31) ;
   \draw[onearrow] ($($(l21)!.6!(l32)$)$) -- (l32) ;
   \draw[onearrow] ($($(l21)!.6!(l33)$)$) -- (l33) ;
   \draw[onearrow] ($($(l22)!.5!(l32)$)$) -- (l32) ;
   \draw[onearrow] ($($(l22)!.5!(l33)$)$) -- (l33) ;
   \draw[onearrow] ($($(l23)!.5!(l33)$)$) -- (l33) ;

   \node[r] (l43) at (7.2,1.5)  {win for \$3};
   \node[r] (l42) at (7.2,0.5)  {win for \$2};
   \node[r] (l41) at (7.2,-0.5) {win for \$1};
   \node[r] (l40) at (7.2,-1.5) {win for \$0};
   \node[r] (l4l) at (6.8,-2.8)   {lose};

   \draw[onearrow] (l30) -- (l40) ;
   \draw[onearrow] (l31) -- (l41) ;
   \draw[onearrow] (l32) -- (l42) ;
   \draw[onearrow] (l33) -- (l43) ;

   \draw[onearrow] (l30) -- (l4l) ;
   \draw[onearrow] (l31) -- (l4l) ;
   \draw[onearrow] (l32) -- (l4l) ;
   \draw[onearrow] (l33) -- (l4l) ;
\end{tikzpicture}
\end{minipage}
\hfill
\begin{minipage}{0.4\textwidth}
  \vspace{-0.4em}
  \caption[Extensive-form Mechanism]{
  An extensive-form menu description for bidder $n$ in a second-price auction
  }
  \label{fig:full-menu-only-2pa}

  {\footnotesize \textbf{Note:}
  The second-to-last layer is labeled with bidder $n$'s menu, abbreviated in the figure by the price she must pay to win the item. \par  }
  \end{minipage}

\end{figure}

More broadly, any precise algorithm taking players types as inputs
induces an extensive-form description in a natural way: the vertices in layer $j$ are the possible states of the algorithm after querying the types of different players altogether $j$ times. 
In particular, our positive results (\autoref{alg:Phpm} and \autoref{alg:individ-dict-ttc}) correspond to extensive-form descriptions.
The definitions of all our simplicity desiderata %
(\autoref{def:preference-linear} and \autoref{def:pick-an-object} below)
also extend naturally to extensive-form descriptions.
Moreover, the proofs of our impossibility theorems (\autoref{thrm:MTR-DA-LB} and \autoref{prop:pick-an-object} below) hold, mutatis mutandis, for the relevant class of extensive form descriptions.

In addition to providing a self-contained mathematical language for expressing our results, the definition of an extensive-form description allows us to clarify some ways in which our impossibility results are strong. 
Namely, while algorithms are often required to work for any number of players, our impossibility results hold even if one can use a separate extensive-form description for each number of players $n$,
and regardless of the computational complexity of such a description. 
Relatedly, our impossibility results follow from direct combinatorial arguments and do not depend on any complexity-theoretic conjectures such as $P\ne NP$.

\renewcommand{\thesection}{T}

\section{On Additional Descriptions of DA}
\label{sec:main-body-delicate-DA-algs}
\label{sec:delicate-DA-algs}

In this appendix, we present additional findings regarding descriptions of DA.
We examine a broad classification of mechanism descriptions. 
While we uncover additional descriptions of DA, we find that all such uncovered descriptions (beyond the traditional one and \autoref{alg:Phpm}) are unintuitive and convoluted algorithms that are impractical for real-world use.

To motivate our search for additional descriptions of DA, consider the outline of menu-in-outcome descriptions, which provided our highly-useful \autoref{alg:individ-dict-ttc} for TTC.
Our description in \autoref{alg:individ-dict-ttc} satisfies applicant-proposing and linear-memory, that may be regarded  as certain formal simplicity properties.
Our main impossibility theorem (\autoref{thrm:MTR-DA-LB}) shows that applicant-proposing menu-in-outcome description of DA must, in some formal sense, be complex; formally, they cannot be linear-memory.
However, this theorem does not give any impossibility result for menu-in-outcome descriptions of DA which---like our menu description of DA, \autoref{alg:Phpm}---are \emph{institution}-proposing.\footnote{
   We use the term institution-proposing to mean the definition perfectly analogous to applicant-proposing (\autoref{def:preference-linear}), in which sides of the market are interchanged.
} 
Given this, one might still hope for a useful institution-proposing menu-in-outcome description of DA, which might yield an alternative outcome description of DA together with a simple proof of its strategyproofness.

Perhaps surprisingly, in \autoref{sec:HPDA-MTR} we construct a new institution-proposing menu-in-outcome description of DA which is, in fact, linear memory.
Unfortunately, this description is \emph{exceedingly} unintuitive and convoluted.
Indeed, as one can see from the details in \autoref{sec:HPDA-MTR}, 
this description is a highly technical algorithm that requires careful bookkeeping to maintain its linear-memory.
Thus, in contrast to DA's traditional description and our \autoref{alg:Phpm}, this algorithm is impractical for describing DA to real-world participants.

Motivated by the intricacies of the description we uncover in \autoref{sec:HPDA-MTR}, in \autoref{sec:non-local-appendix}, we additionally use an established formal simplicity property to demonstrate a sense in which institution-proposing menu-in-outcome descriptions of DA \emph{must be} convoluted and impractical.
Our linear-memory property used in \autoref{thrm:MTR-DA-LB} does not suffice for this purpose (since our convoluted descriptions in \autoref{sec:delicate-DA-algs} satisfy this flexible property).
Instead, we use the \emph{pick-an-object} simplicity desideratum of \citet{BoH20}.
We prove that institution-proposing menu-in-outcome descriptions for DA cannot be pick-an-object.\footnote{
  In \autoref{sec:non-local-appendix}, we demonstrate more generally that for DA, institution-proposing outcome descriptions---and thus menu-in-outcome descriptions as a special case---cannot be pick-an-object.
} 
Briefly and informally, this means that all such descriptions must learn the match of some applicant $d$ when making queries which seem unrelated to $d$,
showing a precise sense in which such descriptions cannot be simple.
Combined with our main impossibility result (\autoref{sec:mtr-matching}), this shows that one-side-proposing menu-in-outcome descriptions of DA cannot (in appropriate senses) be simple.

More broadly, in pursuit of potentially useful descriptions of DA, we consider a broad classification of matching mechanism descriptions.
We consider applicant-proposing descriptions (like traditional ones), and institution-proposing descriptions (like \autoref{alg:Phpm}).
We consider our three description outlines: menu descriptions, outcome descriptions, and menu-in-outcome descriptions.
Altogether, this gives six classes of one-side-proposing descriptions with one of these outlines.
In this appendix, we construct linear-memory descriptions of DA of \emph{every} class that is not ruled out by our main impossibility result \autoref{thrm:MTR-DA-LB}. 
Unfortunately, all of the additional descriptions %
are (like our institution-proposing menu-in-outcome description) exceedingly unintuitive and convoluted algorithms.
See \autoref{fig:All-Delicate-DA-Algs} for an overview of all our descriptions and results for DA.

\begin{table}[htbp]
  \caption{Classification of descriptions of DA}
  \label{fig:All-Delicate-DA-Algs}
  \begin{center}
  \resizebox{\textwidth}{!}{
  {\small
  \begin{tabular}{ccccc}
  \toprule
    & %
    & \makecell{Menu
        \\ Description} 
    & %
        \makecell{Outcome
        \\ Description} 
    & \makecell{Menu-in-Outcome 
        \\ Description} 
    \\ \cmidrule{3-5}
    \\[-0.5em]
    \makecell{Applicant
        \\proposing}
    & %
    & %
        \makecell{Unintuitive, convoluted 
        \\ algorithm in \autoref{sec:DPDA-Menu}.  }
    & %
        \makecell{
        Traditional  
        \\ DA algorithm.}
    & %
        \makecell{ \ul{\textbf{Impossible}} (without 
        \\ quadratic memory) 
        \\ by \autoref{thrm:MTR-DA-LB}. 
        }
    \\ 
    \\[-0.5em]
    \makecell{Institution
        \\proposing}
    & %
    & \multicolumn{1}{c}{
        \makecell{\autoref{alg:Phpm} 
        \\ in \autoref{sec:parallen-hpda}. }
        }
    & %
        \makecell{Unintuitive, convoluted
        \\ algorithm in \autoref{sec:DPDA-AKL-Alg}. }
    & \makecell{Unintuitive, convoluted
        \\ (e.g., not pick-an-object)
        \\ algorithm in \autoref{sec:HPDA-MTR}. }
    \\[-0.5em]
    \\ \bottomrule
  \end{tabular}
  }
  }
  \end{center}
    {\footnotesize \textbf{Notes:}
    We consider descriptions which either read preferences in an applicant-proposing manner or read priorities in an institution-proposing manner.
    We consider three description outlines: menu descriptions (conveying strategyproofness), outcome descriptions (conveying the fully matching), or menu-in-outcome descriptions (conveying both). \par }
\end{table}

All told, our results exhaustively consider all classes of descriptions of DA that are one-side-proposing and fit one of our three description outlines. 
Within this classification, we find two simple and practical descriptions of DA: the traditional one, and our menu description.
This suggests that within our framework, simple descriptions of DA face a trade-off between conveying strategyproofness and conveying the full outcome matching.

The organization of this appendix is as follows.
We present an institution-proposing outcome description of DA, adapted from \citet{AshlagiKL17}, in \autoref{sec:DPDA-AKL-Alg}.
We present our applicant-proposing menu description in \autoref{sec:DPDA-Menu}.
We present our institution-proposing menu-in-outcome description in \autoref{sec:HPDA-MTR}.
We present our supplemental impossibility theorem for DA in \autoref{sec:non-local-appendix}.

\subsection{Institution-proposing outcome description of DA}
\label{sec:DPDA-AKL-Alg}

First, we construct an \emph{institution}-proposing linear-memory outcome description of DA.
Interestingly, essentially this same algorithm was used as a lemma by
\cite{AshlagiKL17} (henceforth, AKL).\footnote{
  For context, \cite{AshlagiKL17} needs such an algorithm to analyze (for a
  random matching market) the expected ``gap'' between the applicant and
  institution optimal stable matching.
  Their algorithm builds on the work of \cite{ImmorlicaM05}, and is also
  conceptually similar to algorithms for constructing the ``rotation
  poset'' in a stable matching instance \cite{GusfieldI89}
  (see also \cite{CaiT19}).
}
For notational convenience, throughout the rest of this appendix, we refer to the priorities of institutions as ``preferences.'' 
We also denote the set of applicants by $\D$, the set of institutions by $\H$, and (when relevant) we describe the menu to applicant $d_*$.

\newcommand{\Hrej}{\H_{\mathrm{*}}}
\newcommand{\Dterm}{\D_{\mathrm{term}}}
\newcommand{\muSure}{\mu}

\begin{theorem}[Adapted from \citealp{AshlagiKL17}]
  \autoref{alg:AKL-alg} computes the applicant-optimal stable outcome.
  Moreover, \autoref{alg:AKL-alg} is an %
  institution-proposing and
  $\widetilde \BigO(n)$-memory description.
\end{theorem}
\begin{proof}
  AKL refer to the sides of the market as ``men'' and ``women'', and define
  ``Algorithm 2 (MOSM to WOSM)'', a men-proposing algorithm for the
  women-optimal stable matching.
  \autoref{alg:AKL-alg} follows the exact same order of proposals as this algorithm from AKL. The only
  difference apart from rewriting the algorithm in a more ``pseudocode''
  fashion is that \autoref{alg:AKL-alg} performs bookkeeping in a
  slightly different way---Algorithm 2 from AKL maintains
  \emph{two} matchings, and their list $V$ keeps track of only women
  along a rejection chain; our list $V$ keeps track of both applicants and
  institutions along the rejection chain (and can thus keep track of the
  ``difference between'' the two matchings which AKL tracks).

  Moreover, the algorithm is institution-proposing, by construction.
  Furthermore, as it runs it stores only a single matching $\muSure$, a set
  $\Dterm \subseteq \D$, and
  the ``rejection chain'' $V$ (which can contain each applicant $d\in\D$
  \emph{at most once}). Thus, it uses memory $\widetilde O(n)$.
\end{proof}

\begin{algorithm}
  \caption{An institution-proposing outcome description of  DA}
    \label{alg:AKL-alg}

    {\footnotesize    
  \begin{flushleft}
    \textbf{Input:}
    Preferences of all applicants $\D$ and institutions $\H$ \\

    \textbf{Output:}
    The result of applicant-proposing deferred acceptance
    \end{flushleft}
\begin{algorithmic}[1]
  \LComment{We start from the institution-optimal outcome, and slowly 
    ``improve the match for the applicants''}
  \State Let $\muSure$ be the result of institution-proposing DA
  \State Let $\Dterm$ be all applicants unmatched in $\muSure$
    \Comment{$\Dterm$ is all applicants at their optimal stable partner}

  \While { $\Dterm \ne \D$ }
    \label{line:outermostPickDMystery}
    \State Pick any $\widehat d\in \D\setminus \Dterm$, and set 
      $d = \widehat d$
      \State Let $h = \muSure(d)$ and set $V = [ (d, h) ]$
    \While{$V \ne []$ }
      \State Let $d \leftarrow \Call{NextAcceptingApplicant}{\muSure,h}$
      \If {$d = \emptyset$ or $d \in \Dterm$}
        \LComment{In this case, all the applicants in $V$ have reached
          their optimal stable partner.}
        \State Add every applicant which currently appears in $V$ to $\Dterm$
        \State Set $V = []$
      \ElsIf {$d \ne \emptyset$ and $d$ does not already appear in $V$}
        \Comment {Record this in the rejection chain}
        \State Add $(d,\muSure(d))$ to the end of $V$
        \State Set $h \leftarrow \muSure(d)$ 
          \Comment{The next proposing institution will be the ``old match'' of $d$.}
      \ElsIf {$d \ne \emptyset$ and $d$ appears in $V$}
        \LComment{A new ``rejection rotation'' should be written to $\muSure$}
        \State \Call{WriteRotation}{$\muSure$, $V$, $d$, $h$}
          \Comment{Updates the value of $\muSure$, $V$, and (possibly) $h$}
      \EndIf
    \EndWhile

  \EndWhile
  \State \textbf{Return} $\muSure$
  \\

  \Function{NextAcceptingApplicant}{$\mu$, $h$}
    \Repeat
      \State \textbf{Query} $h$'s preference list to get their next choice $d$
    \Until{ $d = \emptyset$ or $h \succ_d \mu(d)$ }
    \State \textbf{Return} d 
  \EndFunction
  \vspace{0.15in}
  \Procedure{WriteRotation}{$\muSure$, $V$, $d$, $h$}
    \State Let $T = (d_1,h_1),\ldots,(d_k,h_k)$ be the suffix of $V$
      starting with the first occurrence of $d = d_1$
    \State Update $\muSure$ such that $\muSure(h_i) = d_{i+1}$
      (for each $i=1,\ldots,k$, with indices taken mod $k$)
    \LComment{Now we fix $V$ and $h$ to reflect the new $\muSure$}
    \State Update $V$ by removing $T$ from the end of $V$
    \If { $V \ne \emptyset$ }
      \State Let $(d_0, h_0)$ denote the final entry remaining in $V$
      \LComment{The next proposing institution will either $h_k$ or
        $h_0$, depending on which $d_1$ prefers}
        \If { $h_k \succ_{d_1} h_0$ }
          \State Set $h \leftarrow h_0$
        \ElsIf { $h_0 \succ_{d_1} h_k$ }
          \State Add $(d_1, h_k)$ to the end of $V$
          \State Set $h \leftarrow h_k$
        \EndIf
    \EndIf
  \EndProcedure

\end{algorithmic}
}

\end{algorithm}

\subsection{Applicant-proposing menu description of DA}
\label{sec:DPDA-Menu}

\newcommand{\Hterm}{\H_{\mathrm{term}}}
\newcommand{\Hpicky}{\H_{\mathrm{picky}}}
\newcommand{\Hmenu}{\H_{\mathrm{menu}}}

\newcommand{\dtry}{d^{\mathrm{try}}}
\newcommand{\dfail}{d^{\mathrm{fail}}}
\newcommand{\htry}{h^{\mathrm{try}}}
\newcommand{\hfail}{h^{\mathrm{fail}}}

In this section, we construct an applicant-proposing linear-memory menu description of DA.
On an intuitive level, the algorithm works as per the ``brute-force'' menu description in
\autoref{ex:generic-mtr-matching}, but
avoiding the need to ``restart many times'' by using the various 
properties of DA and by careful
bookkeeping (to intuitively ``simulate all of the separate runs of the
brute-force description on top of each other'').

On a formal level, we describe the algorithm as a variant of
\autoref{alg:AKL-alg}.
The proof constructing this algorithm uses a bijection between 
one applicant's menu in DA under some preferences, and some data concerning the 
\emph{institution}-optimal stable
matching under a related set of preferences. 
Our applicant-proposing menu description is then phrased as a variation of
\autoref{alg:AKL-alg},
which (reversing the roles of applicants and institutions from the
presentation in \autoref{alg:AKL-alg}) is able to compute the
institution-optimal matching using an applicant-proposing algorithm.

Fix an applicant $d_*$ and set $P$ that contains
(1) the preferences of all applicants
$\D\setminus\{d_*\}$ \emph{other than $d_*$} over $\H$ 
and (2) the preferences of all institutions $\H$ over all applicants $\D$
(including $d_*$).
We now define the ``related set of preferences'' mentioned above.
Define the \emph{augmented preference list} $P'$ as follows:
For each $h_i \in \H$, we create two additional applicants 
$\dtry_i, \dfail_i$ and two additional institutions $\htry_i, \hfail_i$.
The entire preference lists of these additional agents in $P'$ are as follows:
for each $h_i \in \H$: 
\begin{align*}
  \dtry_i\  & : \ \htry_i \succ h_i \succ \hfail_i
  &&& 
  \dfail_i\  & : \ \hfail_i \succ \htry_i
  \\
  \htry_i\  & : \  \dfail_i \succ \dtry_i
  &&&
  \hfail_i\ & : \ \dtry_i \succ \dfail_i
\end{align*}
We need to modify the preference lists of the pre-existing institutions as well.
But this modification is simple:
for each $h_i \in \H$, replace $d_*$ with $\dtry_i$.
The institution-optimal matching for this augmented set of preferences $P'$
will encode the menu, as we need.\footnote{
  For the reader familiar with the rotation poset of stable matchings \citep{GusfieldI89}, the
  intuition for this construction is the following: having $\htry_i$
  reject applicant $\dtry_i$ corresponds to $d_*$ ``trying'' to get $h_i
  \in \H$, i.e., ``trying to see if $h_i$ is on their menu.''
  If $d_*$ would be rejected by $h_i$ after proposing,
  either immediately or after some ``rejection rotation,'' then so will
  $\dtry_i$ (because they serve the same role as $d_*$ at $h_i$). 
  So if a rotation swapping $\htry_i$ and $\hfail_i$ exists (e.g., in the
  institution optimal matching) then $h_i$ is \emph{not} on $d_*$'s menu.
  On the other hand, if $d_*$ could actually permanently match to $h_i$, 
  then $\dtry_i$ proposing to $h_i$ will result in a rejection chain that
  ends at some other applicant (either exhausting their preference list or
  proposing to an institution in $\Hterm$), which does not result in
  finding a rotation (or writing a new set of matches as we ``work towards
  the institution-optimal match''). Thus, if 
  $\htry_i$ and $\hfail_i$ do not swap their matches in the
  institution-optimal stable outcome, then $h_i$ \emph{is} on $d_*$'s menu.
}

\begin{proposition}
  An institution $h_i \in \H$ is on $d_*$'s menu in APDA with
  preferences $P$ 
  if and only if
  in the institution-optimal stable matching with the augmented preferences
  $P'$, we have $\htry_i$ matched to $\dtry_i$.
\end{proposition}
\begin{proof}
  For both directions of this proof, we use the following
  lemma, which is a special case of the main technical lemma 
  in~\cite{CaiT22}:
  \begin{lemma}
    In $P'$, each $\htry_i$ has a unique stable partner if and only if, when
    $\htry_i$ rejects $\dtry_i$ (i.e. if $\htry_i$ submitted a list
    containing only $\dfail_i$, and all other preferences remained the same),
    $\htry_i$ goes unmatched (say, in the applicant-optimal matching).
  \end{lemma}
  Note that each $\htry_i$ is matched to $\dtry_i$ in the applicant-optimal
  matching with preferences $P'$ (and the matching among all original
  applicants and institutions is the same as $\mu_{\mathrm{app}}$).

  ($\Leftarrow$) By the lemma, if $\htry_i$ is matched to $\dtry_i$ in the
  institution-optimal matching under $P'$, then $\htry_i$ must go unmatched when
  $\htry_i$ rejects $\dtry_i$. But, after $\htry_i$, we know
  $\dtry_i$ will propose to $h_i$, and some rejection chain may be started.
  Because $\dtry_i$'s very next choice is $\hfail_i$ (and proposing there
  would lead directly to $\htry_i$ receiving a proposal from $\dfail_i$), 
  the \emph{only} way for $\htry_i$ to remain unmatched is if $\dtry_i$
  remains matched to $h_i$. But because (relative to all the original
  applicants) $\dtry_i$ is in the same place as $d_*$ on $h_i$'s preference
  list, the resulting set of rejections in $P'$ will be precisely the same as those
  resulting from $d_*$ submitting a preference list in $P$ which contains
  only $h_i$. In particular, $d_*$ would remain matched at $h_i$ in $P$ if
  they submitted such a list. Thus, $h_i$ is on $d_*$'s menu.

  ($\Rightarrow$) Suppose $\htry_i$ is matched to $\dfail_i$ in the
  institution optimal matching under $P'$. Again, $\htry_i$ must receive a
  proposal from $\dfail_i$ when $\htry_i$ rejects $\dtry_i$. But this can
  only happen if $\dtry_i$ is rejected by $h_i$ (then proposes to
  $\hfail_i$). But because the preferences of the original applicants in $P'$
  exactly corresponds to those in $P$, we know that $d_*$ would get rejected by $h_i$ if they proposed to them in $\mu_{\mathrm{app}}$ under
  $P$. But then $h_i$ cannot be on $d_*$'s menu.
\end{proof}

With this lemma in hand, we can now show that there is an applicant-proposing linear-memory menu description of (applicant-optimal) DA. This description is given in \autoref{alg:dpdaMenu}.
\begin{algorithm}
  \caption{An applicant-proposing menu description of DA}
    \label{alg:dpdaMenu}

  \begin{flushleft}
    \textbf{Input:}
    An applicant $d_*$ and preferences of all applicants $\D\setminus \{d_*\}$ and institutions $\H$ \\

    \textbf{Output:}
    The menu of $d_*$ in applicant-optimal DA given these preferences
    \end{flushleft}
\begin{algorithmic}[1]
\State Simulate the flipped-side version of \autoref{alg:AKL-alg} (such that applicants propose) on preferences $P'$ to get a matching $\mu$
\State \textbf{Return} the set of all institutions $h_i$ such that $\htry_i$ is matched to $\dtry_i$ in $\mu$
\end{algorithmic}
\end{algorithm}
\begin{theorem}
  There is an applicant-proposing, $\widetilde O(n)$ memory menu
  description of (applicant-optimal) DA.
\end{theorem}
\begin{proof}
  The algorithm proceeds by simulating a run of \autoref{alg:AKL-alg}
  on preferences $P'$ (interchanging the role of applicants and
  institutions, so that applicants are proposing).
  This is easy to do while still maintaining the applicant-proposing and
  $\widetilde O(n)$ memory.
  In particular, $P'$ adds only $O(n)$ applicants and institutions, with
  each $\dtry_i$ and $\dfail_i$ making a predictable set of proposals.
  Moreover, the modification made to the preferences lists of the institutions $h \in
  \H$ is immaterial---when such institutions receive a proposal from
  $\dtry_i$, the algorithm can just query their lists for $d_*$.
\end{proof}

\subsection{Institution-proposing menu-in-outcome description of DA}
\label{sec:HPDA-MTR}
\newcommand{\Pnorej}{P_{\mathrm{hold}}}
\newcommand{\RejDag}{\Delta}
\newcommand{\dhold}{d^{\text{hold}}}

In this section, we construct an institution-proposing linear-memory menu-in-outcome description of DA.\footnote{
  For some technical intuition on why such a description might exist,
  consider the construction used in \autoref{thrm:MTR-DA-LB}, and consider
  a menu-in-outcome description for applicant $i$ executed on these
  preferences. To find the menu in this construction with an
  applicant-proposing algorithm, all of the ``top tier rotations'' must be
  ``rotated'', but to find the correct final matching after learning $t_i$,
  some arbitrary subset of the rotations must be ``unrolled'' (leaving only
  the subset of rotations which $t_i$ actually proposes to).
  \autoref{thrm:MTR-DA-LB} shows that all of this information must thus be
  remembered in full. Now consider a run of \autoref{alg:AKL-alg} on these
  preferences (or on a modified form of these preferences where
  institutions' preference lists determine which top tier rotations propose
  to bottom tier rotations). Some subset of top-tier institutions will
  propose to applicant $i$. To continue on with a run of
  \autoref{alg:AKL-alg}, it suffices to undo \emph{exactly one} of
  these proposals. So, if two or more top-tier rotations trigger a
  bottom-tier rotation, then we can be certain that the bottom-tier
  rotation will be rotated, and we only have to remember which bottom-tier
  rotations are triggered by {exactly one} top-tier rotation (which takes
  $\widetilde O(n)$ bits).
}
Throughout this section, let $P|_{d_i : L}$ denote altering preferences $P$
by having $d_i$ submit list $L$.

Unlike our applicant-proposing menu description of DA from
\autoref{sec:DPDA-Menu}, our institution-proposing menu-in-outcome
description cannot be ``reduced to'' another algorithm
such as \autoref{alg:AKL-alg}. However, the algorithm is indeed a modified version of
\autoref{alg:AKL-alg} that ``embeds'' our simple institution-proposing
menu algorithm \autoref{alg:Phpm} (i.e., IPDA where an applicant $d_*$
submits an empty preference list) as the ``first
phase.'' The key difficulty the algorithm must overcome is being able to ``undo
one of the rejections'' made in the embedded run of \autoref{alg:Phpm}.
Namely, the algorithm must match $d_*$ to her top choice from her menu,
and ``undo'' all the rejections caused by $d_*$ rejecting her 
choice.\footnote{
  \autoref{alg:AKL-alg} is independent of the order in which
  proposals are made. Moreover, one can even show that $d_*$ receives 
  proposals from all $h$ on her menu in \autoref{alg:AKL-alg}.
  However, this does not suffice to construct our menu-in-outcome description
  simply by changing the order of \autoref{alg:AKL-alg}. 
  The main reason is this:
  in \autoref{alg:AKL-alg}, the preferences of $d_*$ are already known, so
  $d_*$ can reject low-ranked proposals without remembering the effect that
  accepting their proposal might have on the matching.
  While the ``unrolling'' approach of \autoref{alg:hpda-mtr} is inspired
  by the way \autoref{alg:AKL-alg} effectively ``unrolls rejection chains''
  (by storing rejections in a list $V$ and only writing these rejections to
  $\mu$ when it is sure they will not be ``unrolled''), 
  the bookkeeping of \autoref{alg:hpda-mtr} is far more complicated
  (in particular, the description maintains a DAG $\RejDag$ instead of a list
  $V$).
}
To facilitate this, the description has $d_*$ reject institutions that propose
to $d_*$``as slowly as possible,'' and maintains a delicate 
$\widetilde O(n)$-bit data structure that allows it to undo one of $d_*$'s
rejections.\footnote{
  Interestingly, this ``rolled back state'' is \emph{not} the result of
  institution-proposing DA on preferences $(P, {d_i : \{h_j\}})$, where $h_j$
  is $d_i$'s favorite institution on her menu.
  Instead, it is a ``partial state'' of \autoref{alg:AKL-alg} (when run on
  these preferences),
  which (informally) may perform additional ``applicant-improving
  rotations'' on top of the result, 
  and thus we can continue running 
  \autoref{alg:AKL-alg} until we find the applicant-optimal outcome.
}
The way this data structure works is involved, 
but one simple feature that illustrates how and why it works is the following:
\emph{exactly one} rejection from $d_*$ will be undone, so if some event is
caused by \emph{more than one} (independent) rejection from $d_*$, then
this event will be caused regardless of what $d_*$ picks from the menu.

We present our algorithm in \autoref{alg:hpda-mtr}.
For notational convenience, we define a related set of preferences
$\Pnorej$ as follows:
For each $h_i \in \H$, add a ``copy of $d_*$''
called $d^{\text{hold}}_i$ to $\Pnorej$. The only acceptable institution
for $d^{\text{hold}}_i$ is $h_i$, and if $d_*$ is on
$h_i$'s list, replace $d_*$ with $d^{\text{hold}}_i$ on $h_i$'s list.
Given what we know from \autoref{sec:parallen-hpda}, the proof that this
algorithm calculates the menu is actually fairly simple:

\newcommand{\muLb}{\mu}

\begin{algorithm}
  \caption{An institution-proposing menu-in-outcome description of
     DA}\label{alg:hpda-mtr}
     {\footnotesize
  \begin{flushleft}
    \textbf{Phase 1 input:}
    An applicant $d_*$ and preferences of applicants $\D \setminus \{d_*\}$ and
    institutions $\H$
 
    \textbf{Phase 1 output:}
    The menu $\H_{\mathrm{menu}}$ presented to $d_*$ in (applicant-proposing) DA

    \textbf{Phase 2 input:}
    The preference list of applicant $d_*$

    \textbf{Phase 2 output:}
    The result of (applicant-proposing) DA
    \end{flushleft}
\begin{algorithmic}[1]
  \LComment{\textbf{Phase 1:} }
  \State Simulate a run of $ IPDA(\Pnorej) $ and call the result $\mu'$
  \State Let $\Hrej$ be all those institutions $h_i \in \H$
    matched to $\dhold_i$ in $\muLb'$
    \Comment{These institutions ``currently sit at $d_*$''}
  \State Let $\mu$ be $\mu'$, ignoring all matches of the form $(\dhold_i, h)$
  \State Let $\Hmenu$ be a copy of $\Hrej$
    \Comment{We will grow $\Hmenu$}
  \State Let $\RejDag$ be an empty graph
  \Comment{The ``unroll DAG''. After Phase 1, we'll ``unroll a chain of
    rejections''}

  \While{ $\Hrej \ne \emptyset$ }
    \State Pick some $h \in \Hrej$ and remove $h$ from $\Hrej$
    \label{line:HrejPick}
    \State Add $(d_*, h)$ to $\RejDag$ as a source node
    \State Set $P = \{ (d_*, h) \}$ \Comment{This set stores the ``predecessors
      of the next rejection''}
    \While { $h \ne \emptyset$ }
      \State Let $d \leftarrow \Call{NextInterestedApplicant}{\muLb,\RejDag, h}$

      \Call{AdjustUnrollDag}{$\muLb$, $\RejDag$, $P$, $d$, $h$}
        \Comment{Updates each of these values}

    \EndWhile
  \EndWhile

  \State \textbf{Return} $\Hmenu$
  \LComment{ \textbf{Phase 2:} We now additionally have access to $d_*$'s preferences }
  \State Permanently match $d_*$ to their top pick $h_{\mathrm{pick}}$ from $\Hmenu$
  \State $(\mu, \Dterm) \leftarrow$ 
    \Call{UnrollOneChain}{$\mu, \RejDag, h_{\mathrm{pick}}$}
  \State \textbf{Continue} running the \autoref{alg:AKL-alg}
  until its end, using this $\mu$ and $\Dterm$, starting
  from \autoref{line:outermostPickDMystery}
  \State \textbf{Return} the matching resulting from \autoref{alg:AKL-alg}
\\
\vspace{0.15in}
\Function{NextInterestedApplicant}{$\mu$, $\RejDag$, $h$}
  \Repeat
    \State \textbf{Query} $h$'s preference list to get their next choice $d$
  \Until{ $d \in \{ \emptyset, d_* \}$ OR ($d$ is in $\RejDag$, paired with $h'$ in
    $\RejDag$, and $h \succ_d h'$) OR ($d$ is not in $\RejDag$ and $h
    \succ_d \mu(d)$) }
  \State \textbf{Return} d 
\EndFunction
\vspace{0.15in}
\Procedure{UnrollOneChain}{$\mu$, $\RejDag$, $h_{\text{pick}}$}
  \State Let $(d_0, h_0), (d_1, h_1), \ldots, (d_k, h_k)$ be the
    (unique) longest chain in $\RejDag$ starting from 
    $(d_0, h_0) = (d^*, h_{\mathrm{pick}})$
  \State Set $\muLb(d_i) = h_i$ for $i=0,\ldots,k$
  \State Set $\Dterm = \{d_*, d_1, \ldots, d_k\}$
  \State \Return $(\mu, \Dterm)$
\EndProcedure
\end{algorithmic}
}
\end{algorithm}

\begin{algorithm}
{\footnotesize
\begin{algorithmic}[1]
\Procedure{AdjustUnrollDag}{$\muLb$, $\RejDag$, P, $d$, $h$}
    \If {$d = \emptyset$} 
       \State Set $h = \emptyset$
        \Comment{Continue and pick a new $h$}
    \ElsIf { $d = d^*$ } \Comment{$h$ proposes to $d_*$, so we've
      found a new $h$ in the menu}
    \State Add $h$ to $\Hmenu$
    \State Add $(d_*, h)$ to $\RejDag$
    \State Add $(d_*, h)$ to the set $P$ \Comment{$h$ still proposes;
      the next rejection will have multiple predecessors}
  \ElsIf {$d$ does not already appear in $\RejDag$}
    \Comment {Here $h \succ_d \mu(d)$}
    \label{line:RejDagNoCollision}
    \State Add $(d,\muLb(d))$ to $\RejDag$
      \Comment{Record this in the rejection DAG}
    \State Add an edge from each $p\in P$ to $(d, \muLb(d))$ in
      $\RejDag$, and set $P = \{ (d, \muLb(d)) \} $
    \State Set $h' \leftarrow \muLb(d)$, then $\muLb(d) \leftarrow h$, 
      then $h \leftarrow h'$
      \label{line:ProposerSwitch}
    \LComment{The next proposing institution will be the ``old match'' of $d$.}
  \ElsIf {$d$ appears in $\RejDag$}
    \State \Call{AdjustUnrollDagCollision}{$\muLb$, $\RejDag$, $P$, $d$, $h$}
      \Comment{Updates each of these values}
      \label{line:EndOfUpdates}
  \EndIf
\EndProcedure
\vspace{0.15in}
\Procedure{AdjustUnrollDagCollision}{$\muLb$, $\RejDag$, P, $d$, $h$}
  \State Let $p_1 = (d_1, h_1)$ be the pair where $d = d_1$ appears in $\RejDag$
    \Comment{We know $h \succ_{d_1} h_1$}
  \State Let $P_1$ be the set of all predecessors of $p_1$ in $\RejDag$

  \item[]
  \LComment{First, we drop all rejections from $\RejDag$ which we are now
  sure we won't have to unroll}
    \State Let $(d_1,h_1),\ldots,(d_k,h_k)$ be the (unique) longest
      possible chain in $\RejDag$ starting from $(d_1, h_1)$ \\
      \phantom{.}\qquad such that \emph{each node $(d_j,h_j)$ for $j>1$ has exactly one predecessor}
      \label{line:removeNodes}
  \State Remove each $(d_i, h_i)$ from $\RejDag$, for $i=1,\ldots,k$, and
    remove all edges pointing to these nodes

  \item[]
  \LComment{Now, we adjust the nodes to correctly handle $d_1$ 
  (which might have to ``unroll to $h_{\text{min}}$'')}
  \State Let $h_{\text{min}}$ be the institution among $\{\mu(d_1), h\}$ which
    $d_1$ prefers least
    \label{line:newCollisionNodes}
  \State Let $p_{\text{new}} = (d_1, h_{\text{min}})$; add $p_{\text{new}}$
    to $\RejDag$
  \If{  $h_{\text{min}} = h$ } \Comment{We replace $p_1$ with $p_{\text{new}}$}
    \State Add an edge from each $p \in P_1$ to $p_{\text{new}}$
    \label{line:AddEdgePredException}
    \State Add $p_{\text{new}}$ to $P$ 
      \Comment{$h$ is still going to propose next}
  \Else \Comment{  Here $h_{\text{min}} = \mu(d_1)$; we add
    $p_{\text{new}}$ below the predecessors $P$ }
    \State Add an edge from every $p \in P$ to $p_{\text{new}}$
    \State Set $P = P_1 \cup \{p_{\text{new}}\}$
    \State Set $h' \leftarrow \muLb(d_1)$, then $\muLb(d) \leftarrow h$, 
      then $h \leftarrow h'$
      \Comment{$d_1$'s old match will propose next}
      \label{line:ProposerSwitchCyclic}
  \EndIf

\EndProcedure
\end{algorithmic}
}
\end{algorithm}

\begin{lemma}
  The set $\Hmenu$ output by \autoref{alg:hpda-mtr} is the menu of $d_*$ in
  (applicant-proposing) DA.
\end{lemma}
\begin{proof}
  Ignoring all bookkeeping, Phase 1 of this algorithm corresponds to a run of
  $IPDA(P|_{d_* : \emptyset})$.
  The only thing changed
  is the order in which $d_*$ performs rejections, but
  DA is invariant under the order in which rejections are performed.
  Moreover, $\Hmenu$ consists of exactly all
  institutions who propose to $d$ during this process,
  i.e. $d_*$'s menu (according to \autoref{sec:parallen-hpda}).
\end{proof}

The correctness of the matching, on the other hand, requires an
involved proof. The main difficulty surrounds the ``unroll DAG'' $\RejDag$,
which must be able to ``undo some of the rejections'' caused by $d_*$
rejecting different $h$.
We start by giving some invariants of the state maintained by the algorithm
(namely, the values of $\RejDag$, $\mu$, $P$, and $h$):

\begin{lemma}
  \label{lem:structureOfRejDag}
  At any point outside of the execution of {\sc AdjustUnrollDAG}:
  \begin{enumerate}[(1)]
  \itemsep0em
  \item $P$ contains all nodes in $\RejDag$ of the form $(d,h)$ (where $h$
    is the ``currently proposing'' $h\in \H$).
  \item All of the nodes in $P$ have out-degree $0$.
  \item The out-degree of every node in $\RejDag$ is at most $1$.
  \item Every source node in $\RejDag$ is of the form $(d_*, h_i)$ for some
    $h_i \in \Hmenu$.
  \item For every edge $(d_0, h_0)$ to $(d_1, h_1)$ in $\RejDag$,
    we have $\mu(d_1) = h_0$.
  \item For each $d \in \D\setminus \{d_*\}$, there is at most one node in $\RejDag$ of the form $(d, h_i)$ for some $h_i$.
  \end{enumerate}
\end{lemma}
Each of these properties holds trivially at the beginning of the algorithm, and it is straightforward to verify that each structural property is maintained each time {\sc AdjustUnrollDag} runs.

We now begin to model the properties that $\RejDag$ needs to maintain as
the algorithm runs.
\begin{definition}
  At some point during the run of any institution-proposing algorithm with
  preferences $Q$,
  define the \emph{truncated revealed preferences} $\overline Q$ as exactly those
  institution preferences which have been queried so far, and assuming that
  all further queries to all institutions will return $\emptyset$ (that is,
  assume that all institution preference lists end right after those
  preferences learned so far).

  For some set of preferences $Q$
  we say the revealed truncated preferences $\overline Q$ and the pair
  $(\mu', \Dterm')$ 
  is a \emph{partial AKL state} for
  preferences $Q$ if there exists some execution order of
  \autoref{alg:AKL-alg} and a point along that execution path
  such that the truncated revealed preferences are $\overline Q$, and
  $\mu$ and $\Dterm$ in \autoref{alg:AKL-alg}
  take the values $\mu'$ and $\Dterm'$ 

  Let $Q$ be a set of preferences which does not include preference of
  $d_*$, and let $\overline Q$ a truncated revealed
  preferences of $Q$.
  Call a pair $(\mu, \RejDag)$ \emph{unroll-correct for $Q$ at 
  $\overline Q$} if
  1) $\mu$ is the result of $IPDA(\overline Q)$, and moreover, for every $h\in \Hmenu$,
  the revealed preferences $\overline Q$ and pair {\sc UnrollOneChain}$(\mu, \RejDag, h)$ 
  is a valid partial AKL state of
  preferences $(\overline Q, d_*: \{h\})$.
\end{definition}

The following is the main technical lemma we need, which inducts on the
total number of proposals made in the algorithm, and shows that
$(\mu, \RejDag)$ remain correct every time the algorithm changes their
value:
\begin{lemma}
  \label{lem:hpda-mtr-main-technical}
  Consider any moment where we query some institution's preferences list
  withing {\sc NextInterestedApplicant} in \autoref{alg:hpda-mtr}.
  Let $h$ be the just-queried institution, let $d$ be the returned applicant,
  and suppose that the truncated revealed preferences before that query are
  $\overline Q$, and fix the current values of $\mu$ and $\RejDag$.
  Suppose that $(\mu, \RejDag)$ are unroll-correct for $Q$ at $\overline Q$. 

  Now let $\overline Q'$ be the revealed preferences after adding $d$
  to $h$'s list, and let $\mu'$ and $\RejDag'$ be the updated version
  of these values after \autoref{alg:hpda-mtr} processes this proposal
  (formally, if {\sc NextInterestedApplicant} returns $d$, fix $\mu'$ and
  $\RejDag'$ to the values of $\mu$ and $\RejDag$ after the algorithm
  finishes running {\sc AdjustUnrollDag}; if {\sc NextInterestedApplicant}
  does not return $d$, set $\mu' = \mu$ and $\RejDag' = \RejDag$).
  Then $(\mu', \RejDag')$ are unroll-correct for $Q$ at $\overline Q'$. 
\end{lemma}
\begin{proof}
  First, observe that if $h$'s next choice is $\emptyset$, then the claim
  is trivially true, because $\overline Q = Q$ (and {\sc AdjustUnrollDAG}
  does not change $\mu$ or $\RejDag$).
  Now suppose $h$'s next choice is $d\ne\emptyset$, but
  is not returned by {\sc NextInterestedApplicant}. 
  This means that: 1) $d \ne d_*$, 2) $\mu(d) \succ_d h$, and 3) either $d$
  does not appear in $\RejDag$, or $d$ does appear in $\RejDag$, in which case $d$
  matched to some $h'$ such that $h' \succ_d h$. 
  Because $(\mu, \RejDag)$ are unroll-correct for $Q$ at $\overline Q$,
  and because \autoref{lem:structureOfRejDag} says that $d$ can appear at
  most once in $\RejDag$,
  the only possible match which $d$ could be unrolled to at truncated
  revealed preferences $\overline Q$ is $h'$
  (formally, if the true complete preferences were $\overline Q$, 
  then for all $h_* \in \Hmenu$, the partial AKL state 
  under preferences $(\overline Q, d : \{h_*\})$ to which we we would unroll
  would match $d$ to either $\mu(d)$ or $h'$).
  But $d$ would not reject $\mu(d)$ in favor of $h$, nor would she reject
  $h'$ in favor of $h$.
  Thus, (for all choices of $h_* \in \Hmenu$) 
  we know $h$ will always be rejected by $d$,
  and $(\mu, \RejDag)$ are already unroll-correct for $Q$ at 
  $\overline Q'$.

  Now, consider a case where 
  $h$'s next proposal $d \ne \emptyset$ is returned by 
  {\sc NextInterestedApplicant}. 
  There are a number of ways in which {\sc AdjustUnrollDAG} may change
  $\RejDag$. We go through these cases.

  First, suppose $d = d_*$. In this case, the menu of $d_*$ in $\overline
  Q'$ contains exactly one more institution than the menu in $\overline Q$,
  namely, institution $h$. Moreover, for any $h_* \in \Hmenu\setminus \{h\}$, 
  the same partial AKL state is valid under both preferences
  $(\overline Q, d : \{h_*\})$ and $(\overline Q', d : \{h_*\})$ (the only
  difference in $(\overline Q', d : \{h_*\})$ is a single additional
  proposal from $h$ to $d_*$, which is rejected; the correct value of
  $\Dterm$ is unchanged). For $h_* = h$, 
  the current matching $\mu$, modified to match $h$ to $d_*$, is a valid
  partial AKL state for $(\overline Q', d : \{h\})$, and this is exactly
  the result of {\sc UnrollOneChain} (with $\Dterm = \{d_*\}$, which is
  correct for preferences $(\overline Q', d : \{h\})$). 
  Thus, (using also the fact from \autoref{lem:structureOfRejDag}
  that $P$ contains all nodes in $\RejDag$ involving $h$),
  each possible result of {\sc UnrollOneChain} is a correct partial
  AKL state for each $(\overline Q', d : \{h_*\})$, so $(\mu',\RejDag')$ is
  unroll-correct for $Q$ at $\overline Q'$.

  Now suppose $d \notin \{\emptyset, d_*\}$ is  returned from 
  {\sc AdjustUnrollDAG}, and $d$ does not already appear in $\RejDag$.
  In this case, $h \succ_d \mu(d)$, and
  for every $h_* \in \Hmenu$, the unrolled state when
  preferences $(\overline Q, d : \{h_*\})$ will pair $d$ to $\mu(d)$.
  Under preferences $(\overline Q', d : \emptyset)$, a single additional proposal will be
  made on top of the proposals of $(\overline Q, d : \emptyset)$, namely,
  $h$ will propose to $d$ and $d$ will reject $\mu(d)$.
  However, \emph{if $h_*$ is such that $h$ is ``unrolled''} 
  (formally, if $h_*$ is such that
  {\sc UnrollOneChain}$(\mu, \RejDag, h_*)$ changes the partner of $h$)
  then $h$ cannot propose to $d$ in $(\overline Q, d : \emptyset)$ (because
  all pairs in $\RejDag$ can only ``unroll'' $h$ to partners before
  $\mu(h)$ on $h$'s list), nor in 
  $(\overline Q', d : \emptyset)$ (because $\overline Q'$ only adds a
  partner to $h$'s list after $\mu(h)$).
  Thus, for all $h_*$ such that $h$ is unrolled, the pair $(d, \mu(d))$
  should be unrolled as well. 
  On the other hand, for all $h_*$ such that $h$ is not unrolled,
  $h$ will propose to $d$ (matched to $d'$), so $d$ will match to $h$ in
  the unrolled-to state.
  This is exactly how $\mu'$ and $\RejDag'$ specify unrolling should go, as
  needed.

  \textbf{(Hardest case: {\sc AdjustUnrollDagCollision}.)}
  We now proceed to the hardest case, where 
  $d \notin \{\emptyset, d_*\}$ is  returned from 
  {\sc AdjustUnrollDag}, and $d$ already appears in $\RejDag$.
  In this case, {\sc AdjustUnrollDagCollision} modifies $\RejDag$.
  Define $p_1$, $P_1$, and $h_{\text{min}}$,
  following the notation of {\sc AdjustUnrollDagCollision}.
  Now consider any $h_* \in \Hmenu$ under preferences $\overline Q$.
  There are several cases of how $h_*$ may interact with the nodes changed
  {\sc AdjustUnrollDagCollision}, so we look at these cases and prove
  correctness. There are two important considerations which we must prove
  correct: first, we consider
  the way that {\sc AdjustUnrollDagCollision} removes nodes from
  $\RejDag$ (starting on \autoref{line:removeNodes}), and second, we
  consider the way that it creates a new node to handle $d$ 
  (starting on \autoref{line:newCollisionNodes}).

  \textbf{(First part of {\sc AdjustUnrollDagCollision}.)}
  We first consider the way {\sc AdjustUnrollDagCollision} removes nodes from
  $\RejDag$. There are several subcases based on $h_*$.
  First, suppose {\sc UnrollOneChain}$(\mu, \RejDag, h_*)$ does not contain
  $p_1$. Then, because {\sc AdjustUnrollDagCollision} only drops $p_1$ and
  nodes only descended through $p_1$, the chain unrolled by 
  {\sc UnrollOneChain}$(\mu', \RejDag', h_*)$ is unchanged until $h$.
  (We will prove below that the behavior when this chain reaches $h$ is
  correct.) Thus, the initial part of this unrolled chain remains correct
  for $Q$ at $\overline Q'$.

  On the other hand, suppose that {\sc UnrollOneChain}$(\mu, \RejDag, h_*)$
  contains $p_1$. There are two sub-cases based on $\RejDag$. First,
  suppose that there exists a pair $p \in P$ in $\RejDag$ such that $p$ is
  a descendent of $p_1$ (i.e. there exists a $p = (d_x, h) \in P$ and a
  path from $p_1$ to $p$ in $\RejDag$).
  In this case, under preferences $\overline Q$, 
  {\sc UnrollOneChain}$(\mu, \RejDag, h_*)$ would unroll to each pair in
  the path starting at $h_*$, which includes $p_1$ and all nodes on the
  path from $p_1$ to $p$.
  Under $\RejDag'$, however, \emph{none of the nodes from $p_1$ to $p$}
  will be unrolled in this case. The reason is this:
  in \autoref{alg:AKL-alg}, the path from $p_1$ to $p$, including the
  proposal of $h$ to $d_1$, form an ``improvement rotation'' when the true
  preferences are $\overline Q'$. Formally, under preferences 
  $(\overline Q', d_* : \{h_*\})$, if $d_1$ rejected $h_1$, the 
  rejections would follow exactly as in the 
  path in $\RejDag$ between $p_1$ and $p$, and finally $h$ would propose to
  $d_1$.
  \autoref{alg:AKL-alg} would then call {\sc WriteRotation}, and
  the value of $\mu$ would be updated for each $d$ on this path.
  So deleting these nodes is correct in this subcase.\footnote{
    This is the core reason why \autoref{alg:hpda-mtr} cannot ``unroll''
    to $IPDA(Q, d: \{h_i\})$---instead, it unrolls to a ``partial state of
    AKL''.
  }

  For the second subcase, suppose that there is \emph{no} path between
  $p_1$ and any $p \in P$ in $\RejDag$.
  In this case, there must be some source $(d_*, \overline h)$ in $\RejDag$
  which is an ancestor of some $p \in P$,
  and such that the path from $(d_*, \overline h)$ to $p$ does not contain
  any descendent of $p_1$. (This follows because each $p \in P$ must have
  at least one source as an ancestor, and no ancestor of any $p\in P$ can
  be descendent of $p_1$.) 
  To complete the proof in this subcase, it suffices to show that 
  at preferences $(\overline Q', d_* : \{h_*\})$, we ``do not need to
  unroll'' the path in $\RejDag$ starting at $h_*$ after $p_1$
  (formally, we want to show that if you unroll from $\mu'$ the path in
  $\RejDag$ from $h_*$ to just before $p_1$ (including the new node added
  by the lines starting on \autoref{line:newCollisionNodes}),
  then this is a partial AKL state of $Q$ at $\overline Q'$).
  The key observation is this: in contrast to preferences 
  $(\overline Q, d_* : \{h_*\})$, where pair $p_1$ is ``unrolled'',
  under preferences $(\overline Q', d_* : \{h_*\})$, we know
  \emph{$h$ will propose to $d_1$ anyway}, because $d_*$ will certainly
  reject $\overline h$ (and trigger a rejection chain leading from $(d_*,
  \overline h)$ to $h$ proposing to $d_1$).

  \textbf{(Second part of {\sc AdjustUnrollDagCollision}.)}
  We now consider the second major task of {\sc AdjustUnrollDagCollision},
  namely, creating a new node to handle $d$.
  The analysis will follow in the same way regardless of how the first part of {\sc AdjustUnrollDagCollision} executed (i.e., regardless of whether there exists a path between $p_1$ and $P$).
  The analysis has several cases.
  First, suppose $(d_*, h_*)$ is not an ancestor of any node in 
  $P_1 \cup P$ in $\RejDag$. 
  This will hold in $\RejDag'$ as well, so neither 
  {\sc UnrollOneChain}$(\mu, \RejDag, h_*)$ nor will
  {\sc UnrollOneChain}$(\mu', \RejDag', h_*)$ will not change the match of
  $d$. Instead, the match of $d$ under
  {\sc UnrollOneChain}$(\mu', \RejDag', h_*)$ will be $\mu'(d)$, which is a
  correct partial AKL state under 
  $(\overline Q', d_* : \{h_*\})$, as desired.

  Second, suppose $h_*$ is such that $(d_*, h_*)$ is an ancestor
  of some node in $P_1$ in $\RejDag$. 
  There are two subcases.
  If $h_{\text{min}} = h$, then we have $\mu(d_1) = \mu'(d_1)$, but
  when {\sc UnrollOneChain}$(\mu', \RejDag', h_*)$ is run,
  we unroll $d_1$ to $h$.
  Correspondingly, in IPDA with preferences $(\overline Q', d_*: \{h_*\})$,
  we know $d_1$ will not receive a proposal from $\mu(d_1)$ (as this match
  is unrolled in $\overline Q$) but $d_1$ will receive a proposal from $h$
  (as this additional proposal happens in $\overline Q'$ but not in
  $\overline Q$,
  regardless of whether
  this happens due to a ``rejection rotation'' of AKL, or simply due to two
  rejection chains causing this proposal, as discussed above),
  which $d_1$ prefers to the unrolled-to match under preferences 
  $\overline Q$. Thus, 
  under preferences $(\overline Q', d_*: \{h_*\})$,
  we know $d_1$ will match to
  $h_{\text{min}} = h$ in a valid partial AKL-state.
  So $(\mu',\RejDag')$ is correct for $\overline Q'$ in this subcase.
  If, on the other hand, $h_{\text{min}} = \mu(d_1)$, then in $\RejDag'$,
  {\sc UnrollOneChain}$(\mu', \RejDag', h_*)$ will
  not contain the new node $p_{\text{new}}$. However, $\mu'(d_1)=h$, and
  we know $d$ would receive a proposal from $h$
  $(\overline Q', d_* : \{h_*\})$, and would accept this proposal. 
  So $(\mu',\RejDag')$ is correct for $\overline Q'$ in this subcase.

  Third and finally, suppose $h_*$ is such that $(d_*, h_*)$ is an ancestor
  of some node in $P$ in $\RejDag$. The logic is similar to the previous
  paragraph, simply reversed.
  Specifically, there are two subcases.
  If $h_{\text{min}} = h$, then when
  preferences are $(\overline Q', d_* : \{h_*\})$,
  then $d_1$ will no longer receive a proposal from $h$, 
  but will still receive a proposal from $\mu(d_1)$.
  So $d_1$ should remain matched to $\mu(d_1)$ during 
  {\sc UnrollOneChain}$(\mu', \RejDag', h_*)$, and
  $(\mu',\RejDag')$ is correct for $\overline Q'$ in this subcase.
  If $h_{\text{min}} = \mu(d_1)$, 
  then $\mu'(d_1) = h$, and
  in $\RejDag'$, {\sc UnrollOneChain}$(\mu', \RejDag', h_*)$ 
  will contain the new node $p_{\text{new}}$, which unrolls $d_1$ to their
  old match $\mu(d_1)$.
  This is correct, because in $\overline Q$,
  according to $\RejDag$, we know $h$ will be unrolled to some previous
  match, and correspondingly, in 
  preferences $(\overline Q', d_*: \{h_*\})$, we know $d_1$ will never
  receive a proposal from $h$.
  So $(\mu',\RejDag')$ is correct for $\overline Q'$ in this subcase.

  Thus, for all cases, $(\mu', \RejDag')$ are unroll-correct for $Q$ at
  $\overline Q'$, as required.
\end{proof}

To begin to wrap up, we bound
the computational resources of the algorithm:
\begin{lemma}
  \autoref{alg:hpda-mtr} is institution-proposing and uses memory
  $\widetilde O(n)$.
\end{lemma}
\begin{proof}
  The institution-proposing property holds by construction. To bound the
  memory, the only thing that we need to consider on top of AKL is the
  ``unroll DAG'' $\RejDag$. 
  This memory requirement is small, because there are at most $O(n)$ nodes
  of the form $(d_*, h)$ for different $h \in \H$, and by 
  \autoref{lem:structureOfRejDag}, a given applicant $d\in
  \D\setminus\{d_*\}$ can appear \emph{at most once} in
  $\RejDag$. So the memory requirement is 
  $\widetilde O(n)$.
\end{proof}

We can now prove our main result:
\begin{theorem}
  \autoref{alg:hpda-mtr} is an institution-proposing,
  $\widetilde O(n)$ memory menu-in-outcome description for DA.
\end{theorem}
\begin{proof}
  We know \autoref{alg:hpda-mtr} correctly computes the menu, and that it
  is institution-proposing and $\widetilde O(n)$ memory.
  So we just need to show that it correctly computes the final matching.
  To do this, it suffices to show that at the end of Phase 1 of
  \autoref{alg:hpda-mtr}, $(\mu, \RejDag)$ is unroll-correct
  for $Q$ at the truncated revealed preferences $\overline Q$ (for then, by
  definition, running \autoref{alg:AKL-alg} after {\sc UnrollOneChain} will
  correctly compute the final matching).

  To see this, first note that an empty graph is unroll-correct for the
  truncated revealed preference after running $IPDA(P_{\text{hold}})$, as
  no further proposals beyond $d_*$ can be made in these truncated
  preferences. Second, each time we pick an $h \in \H_*$
  on \autoref{line:HrejPick}, a single $(d_*, h)$ added to $\RejDag$ (with
  no edges) is unroll-correct for $Q$ at $\overline Q'$, by construction.
  Finally, by \autoref{lem:hpda-mtr-main-technical}, every other query to any
  institution's preference list keeps $(\mu, \RejDag)$ unroll-correct after
  the new query. So by induction, $(\mu, \RejDag)$ is unroll-correct at the
  end of Phase 1, as desired.
\end{proof}

\subsection{Supplemental Impossibility Result for DA}
\label{sec:non-local-appendix}

In this appendix, we give a supplemental impossibility result for descriptions of DA.
We prove that institution-proposing outcome descriptions---and hence institution-proposing menu-in-outcome descriptions as a special case---cannot satisfy the \emph{pick-an-object} simplicity condition of \cite{BoH20}.

\citet{BoH20} introduce the pick-an-object condition in the context of interactive mechanisms, where (informally speaking) agents are iteratively asked to pick their favorite objects from some set, and whenever the mechanism terminates, every agent is matched to their most recently picked object.
For example, in a dynamic mechanism implementing DA, applicants can be iteratively asked to pick their favorite institution from the set of all institutions they have not yet proposed to.
We consider the pick-an-object condition within the context of one-side-proposing outcome descriptions.
In this context, the condition requires that when the description terminates, every agent on the proposing side must be matched to whichever agent they proposed to most recently.

Like the linear-memory condition we use in \autoref{sec:mtr-matching}, the pick-an-object condition captures one feature of matching mechanism descriptions used to explain these mechanisms in practice.
Indeed, %
the description in \autoref{fig:applicant-linear-graphics} on \autopageref{fig:applicant-linear-graphics} that is used by the NRMP is pick-an-object, since the yellow highlighting in that figure tracks the most recent proposal of each applicant and, at the end of the description, relays the outcome matching.
However, where linear-memory is a fairly permissive desideratum concerning the amount of bookkeeping used, pick-an-object is a more restrictive desideratum concerning the manner in which the bookkeeping is updated and used.
Thus, we do not interpret our pick-an-object impossibility result as strongly as our linear-memory impossibility result, e.g., we do not argue that all small tweaks of the traditional description of DA should be pick-an-object.
Nevertheless, our pick-an-object impossibility result is quite useful: It shows a potentially-desirable class of descriptions cannot satisfy an established and intuitive simplicity condition,
and gives a specific barrier that hypothetical more-practical alternatives to our unintuitive and convoluted description in  \autoref{sec:delicate-DA-algs} would have to circumvent.

We now formally define pick-an-object, adapting the definition from \citet{BoH20} to focus on institution-proposing outcome descriptions.

\begin{definition}[Pick-an-Object]
    \label{def:pick-an-object}
   An institution-proposing outcome description is \emph{pick-an-object} if, whenever the description terminates and calculates some outcome matching $\mu$, it satisfies the following. 
   For every institution $h$, let $d_h$ be the most recently queried applicant from $h$'s preference list, i.e., if the description made $j$ queries to $h$, then $d$ is the $j$\textsuperscript{th} applicant on $h$'s preference list.
   Then, $\mu(h)=d_h$ for every institution $h$.
\end{definition}

Observe that the traditional descriptions of SD, TTC, and DA are applicant-proposing outcome descriptions that are pick-an-object (according to a definition perfectly analogous to \autoref{def:pick-an-object}, but interchanging the roles of the applicants and institutions).
DA (the applicant-optimal stable matching mechanism) has a nontrivial institution-proposing outcome description as well (\autoref{sec:DPDA-AKL-Alg}).
However, as we now show, such a description cannot be pick-an-object, giving a sense in which they cannot be simple.
Formally:

\begin{proposition}
\label{prop:pick-an-object}
No institution-proposing outcome description of DA is pick-an-object.
\end{proposition}
\begin{proof}
  Assume for contradiction that $D$ is an institution-proposing
  outcome description of DA which is pick-an-object. 
  Consider a market with institutions $h_1, h_2$ and applicants $d_1,d_2,d_3$.
  We first define preferences of three applicants as follows:
  \begin{align*}
    d_1 & : h_2 \succ h_1 \\
    d_2 & : h_1 \succ h_2 \\
    d_3 & : \text{(any complete preference list)}
  \end{align*}
  Next, we consider two possible preference lists for each of $h_1, h_2$:
  \begin{align*}
    \succ_1 & : d_1 \succ d_2 \succ d_3 &
    \succ_2 & : d_2 \succ d_1 \succ d_3 \\
    \succ_1' & :d_1 \succ d_3 \succ d_2 &
    \succ_2' & : d_2 \succ d_3 \succ d_1
  \end{align*}
  
  One can check that DA (the applicant-optimal stable matching) produces outcome matching $\mu_1$ that assigns $d_1$ to $h_2$ and $d_2$ to $h_1$ when the priorities are $(\succ_1,\succ_2)$; on any other profile of priorities among those defined above, DA has as outcome the matching $\mu_2$ that assigns $d_1$ to $h_1$ and $d_2$ to $h_2$.
  Thus, our description $D$ can know the outcome on these inputs only when it has read the second-highest-priority spot of \emph{both} $h_1$ and $h_2$.
  However, intuitively, this means that our institution-proposing description $D$ of DA cannot be pick-an-object, because the highest-priority applicant for both $h_1$ and $h_2$ must be read before we can know whether these institutions are assigned to these applicants.

  Formally, consider the execution of $D$ when institutions have priorities $(\succ_1, \succ_2)$.
  Consider the final time during this execution when $D$ learns the difference between $\succ_j$ and $\succ_j'$ for some $j \in \{1,2\}$; i.e., %
  the latest possible state $s$ during the execution of the description with priority profile $Q = (\succ_1, \succ_2)$ where the execution 
  diverges from that of some priority profile in $\bigl\{ (\succ_1', \succ_2), (\succ_1, \succ_2')\bigr\}$.
  (Note that the description must learn this difference in order to calculate DA.)
  By the symmetry in the defined preferences, it is without loss of generality to suppose that in state $s$, the description queries the preferences of applicant $1$, and thus has one successor state consistent with $Q$ and another consistent with $Q' = (\succ_1', \succ_2)$.
  However, since $D$ is institution-proposing, this means that in state $s$, the description has already read $d_1$ off the priority list of $h_1$ (and is proceeding to read either $d_2$ or $d_3$ next).
  Since $D$ is pick-an-object, this means that $h_1$ cannot match to $d_1$ in any the final outcome matching of any execution of $D$ consistent with $s$.
  But this is a contradiction, since $h_1$ must match to $d_1$ in DA when the priority profile is $(\succ_1', \succ_2)$. This finishes the proof.
\end{proof}

\autoref{prop:pick-an-object} directly implies that DA has no institution-proposing menu-in-outcome description satisfying the pick-an-object condition (since such a description is, in particular, an outcome description).
Combined with our robust main impossibility result (\autoref{sec:mtr-matching}), this establishes precise impossibilities for simple one-side-proposing menu-in-outcome descriptions of DA:
Such applicant-proposing descriptions cannot be linear-memory, and such institution-proposing descriptions cannot be pick-an-object.

\end{document}